\documentclass[a4paper,USenglish,cleveref,11pt]{article}

\usepackage{xspace,fullpage}
\usepackage{algorithm, algorithmicx}
\usepackage[noend]{algpseudocode}
\usepackage{makecell,adjustbox,xcolor}
\usepackage{float}
\usepackage{amsthm,amssymb,amsmath}
\usepackage{thmtools,thm-restate}
\usepackage{hyperref}

\newcommand{\m}[1]{\ensuremath{\mathcal{#1}}}

\algnewcommand{\IfL}[1]{\State\algorithmicif\ #1\ \algorithmicthen}
\algnewcommand{\EndIfL}{\unskip\ \algorithmicend\ \algorithmicif}
\newcommand{\rd}[1]{\textsc{read#1}}
\newcommand{\wrt}[1]{\textsc{write#1}}
\newcommand{\adt}[1]{\textsc{audit#1}}

\newcommand{\LL}{\textsc{LL}\xspace}
\newcommand{\SC}{\textsc{SC}\xspace}
\newcommand{\llk}{\textsc{ll}}
\newcommand{\scd}{\textsc{sc}}

\newcommand{\adReg}[2]{($#1,#2$)-auditable register}
\newcommand{\adRegs}[2]{($#1,#2$)-auditable registers}

\newcommand{\DenyList}{Deny List\xspace}
\newcommand{\DenyLists}{Deny Lists\xspace}

\title{Auditable Shared Objects:\\ From Registers to Synchronization Primitives}

\author{Hagit Attiya\\
  \small{Technion -- Israel Institute of Technology, Israel}
  \and
  Antonio Fernández Anta\\
  \small{IMDEA Software Institute \& IMDEA Networks Institute, Madrid, Spain}
  \and
  Alessia Milani\\
  \small{Aix Marseille Univ, CNRS, LIS, Marseille, France}
  \and
  Alexandre Rapetti \\
  \small{Université Paris-Saclay, CEA, List, F-91120, Palaiseau, France}
  \and
  Corentin Travers \\
  \small{Aix Marseille Univ, CNRS, LIS, Marseille, France}
}

\newtheorem{theorem}{Theorem}[section]
\newtheorem{lemma}{Lemma}[section]

\newtheorem{proposition}{Proposition}[section]
\newtheorem{definition}{Definition}[section]

\begin{document}
\date{}
\maketitle

   \begin{abstract}
\emph{Auditability} allows to track operations performed on a shared object, recording who accessed which information.
This gives data owners more control on their data.
Initially studied in the context of single-writer registers, 
this work extends the notion of auditability to other shared objects, 
and studies their properties.

We start by moving from single-writer to \emph{multi-writer} registers, 
and provide an implementation of an \emph{auditable 
$n$-writer $m$-reader read / write register},
with $O(n+m)$ step complexity.
This implementation uses $(m+n)$-sliding registers, 
which have consensus number $m+n$.
We show that this consensus number is necessary.
The implementation extends naturally to support an 
\emph{auditable load-linked / store-conditional} (LL/SC) shared object. 
LL/SC is a primitive that supports efficient implementation 
of many shared objects. 
Finally, we relate auditable registers to other access control objects, 
by implementing an \emph{anti-flickering deny list} 
from auditable registers.
\end{abstract}


\setcounter{tocdepth}{2}
{\small{\tableofcontents}}
\thispagestyle{empty}
\newpage
\setcounter{page}{1}
\section{Introduction}

The ability to track operations on shared objects is a fundamental feature in distributed systems. 
\emph{Auditability}, in particular, enables data owners to monitor access to their data by recording who accessed which information. 
This mechanism provides a robust alternative to traditional access control mechanisms, shifting the emphasis from access restriction to post-hoc accountability. 
Auditability is helpful for preserving data privacy, 
as it can be used after a data breach,
to enforce accountability for data access.
This is particularly useful in shared, remotely accessed storage systems, 
where, for instance, understanding the extent of a data breach 
can help mitigate its impact. 

Auditability has been initially studied in the context of non-atomic replicated storage~\cite{BessaniDisc} and single-writer atomic registers~\cite{AttiyaPMPR23}. 
This work explores auditability for \emph{other shared objects}, 
considering the following key questions: 
(1) What types of base objects are necessary and sufficient to 
implement auditable \emph{multi-writer} atomic registers?  
(2) How does auditability impact the design of \emph{common synchronization primitives}?
(3) Can auditable registers be leveraged to construct more complex \emph{access control mechanisms}?
 
We begin to answer these questions by generalizing the notion of auditable atomic registers to the multi-writer setting. 
Hence, we present a wait-free, linearizable implementation of an auditable $n$-writer $m$-reader read/write register (abbreviated as \adReg{m}{n}), with linear (in $n+m$) step complexity.
To achieve these properties, our algorithm uses 
sliding registers~\cite{MostefaouiPR18}: 
a \emph{$k$-sliding register} keeps an ordered list of the latest $k$ values written to it.
If each reader and writer writes at most once in 
an $(m+n)$-sliding register, the sliding register can be used to uniquely order the
operations.
This implies a linearization of the read and write operations 
on the auditable register.
The key challenge is, therefore, how to deal with multiple 
read and write operations on the auditable register.
The backbone of our \adReg{m}{n} implementation is a sequence 
of $(m+n)$-sliding registers, 
each associated with a different value written to the auditable register.
The sliding registers are used to agree which write operation will set the next 
value of the register, and to track which readers have read this value. 
An additional challenge is to be able to efficiently find the ``current'' 
sliding register when reading and writing, to maintain bounded step complexity.

The consensus number of $(m+n)$-sliding registers is exactly $m+n$.
(I.e., they allow to solve consensus among exactly $m+n$ processes.)
Thus, consensus number $m+n$ is sufficient for implementing an \adReg{m}{n}.
We prove that it is also necessary, 
by showing that an \adReg{m}{n} can be used to solve consensus among $m+n$ processes.

Beyond registers, we explore auditability 
for more expressive synchronization primitives, 
like \emph{load-linked/store-conditional} (LL/SC) objects~\cite{jensen1987new}.
The LL operation reads a value from a memory location, while the SC operation writes a new value to the same location only if no process has written to it since the LL operation. If another process has changed the value, the SC fails, requiring the operation to be retried. 
LL/SC enables the efficient construction of 
a wide range of non-blocking data structures~\cite{EllenFKMT2012}.
We leverage the algorithmic ideas of our auditable register
to construct an algorithm that implements an LL/SC object for $n$ processes, using $2n$-sliding register.

Finally, we relate auditable registers to \DenyLists~\cite{FreyGR23}, showcasing how 
auditability can be leveraged to build security primitives. \emph{Allow Lists} and \emph{\DenyLists}, as defined by Frey, Gestin, and Raynal~\cite{FreyGR23}, are two access control objects that give designated processes 
(called \emph{managers}) 
the ability to grant and/or revoke access rights for a given set of
resources. 
Their work specifies and investigates the 
synchronization power of Allow Lists and \DenyLists; 
the latter in two flavors, with and without an \emph{anti-flickering} property. 
(An anti-flickering \DenyList ensures that transient revocations 
do not undermine long-term access control policies.)

We give a relatively simple specification of an \emph{immediate} \DenyList,
which is stronger than the anti-flickering \DenyList of Frey et al., 
and show that it can be efficiently implemented from auditable registers.
The step complexity of the resulting implementation (polynomially) depends on the number of resources and the number of processes. 
In contrast, while being wait-free, the step complexity of the 
anti-flickering \DenyList implementation in~\cite{FreyGR23} 
grows with the number of operations, and is, essentially, unbounded.

\subsection*{Related Work}

Auditability was introduced by Cogo and Bessani~\cite{BessaniDisc} 
in the context of replicated registers. 
They considered a register as an abstraction for distributed 
storage that provides read and write operations to clients. 
Cogo and Bessani define auditability in terms of two properties: 
\emph{completeness} ensures that all readers’ data access are detected,
while \emph{accuracy} ensures that readers who do not access data 
are not wrongly incriminated.
They present an algorithm to implement an auditable \emph{regular} 
(non-atomic)
register, using $n \geq 4f+1$ atomic read/write shared objects,
$f$ of which may be faulty (writers and auditors fail only by crashing; 
faulty readers may be Byzantine).
Their implementation relies on information dispersal schemes, 
where the input of a high-level write is split into several pieces, 
each written in a different low-level shared object.
Each low-level shared object keeps a trace of each access, and in order to read, 
a process has to collect sufficiently many pieces of information in many low-level 
shared objects, which allows to audit the read.

In asynchronous message-passing systems where $f$ processes can be Byzantine,
Del Pozzo, Milani, and Rapetti~\cite{DelPozzoMR2022} study the possibility of implementing
an atomic auditable register,  
with fewer than $4f+1$ servers.
They prove that without communication between servers,
auditability requires at least $4f+1$ servers. 
They also show that allowing servers to communicate with each other admits
an auditable atomic register with optimal resilience of $3f+1$. 

The auditability definition of~\cite{BessaniDisc}  
is tightly coupled with their multi-writer, multi-reader register emulation 
in a replicated storage system using an information-dispersal scheme.
An implementation-agnostic auditability definition was later proposed 
by Attiya et al.~\cite{AttiyaPMPR23}, 
based on collectively linearizing read, write, and audit operations. 
They show that auditing 
adds power to reading and writing, as it allows to solve consensus, 
implying that auditing requires strong synchronization primitives.
They also give several implementations that use non-universal primitives
(like \emph{swap} and \emph{fetch\&add}), for a single writer 
and either several readers or several auditors (but not both). 

Recent work~\cite{AttiyaAFMRT2025} have extended auditability to ensure that even a curious (but honest) reader cannot effectively read data without being audited; in fact, curious readers cannot even audit other readers.
The paper provides implementations of several auditable objects, including registers.
However, the constructions depend on universal primitives, 
including \emph{compare\&swap}, in contrast to our constructions, 
which only employ sliding registers, whose consensus number is bounded. 

Alhajaili and Jhumka~\cite{AlhajailiJ2019} study an interesting variant 
of auditability with malicious processes, with and without a trusted party.
Their definition of auditability has to do with the ability of determining if a system runs as specified, and determine the incorrect behavior in case it does not. This involves systematically tracking and recording system activities to facilitate fault detection and diagnosis. 
The authors propose methodologies for integrating auditability into system design, emphasizing its role in simplifying the debugging process and improving overall system robustness.

Our algorithms (like most prior work) maintain a large amount of auditing information, 
which grows with the number of operations performed in an execution. 
Hajisheykhi, Roohitavaf, and Kulkarni~\cite{HajisheykhiRK2017} investigate how 
to reduce the space used for saving auditing data (to be used to restore the state after a fault), 
by leveraging causal dependencies among them. 
They propose protocols for ensuring accountability in distributed systems when auditable events — deviations from standard protocols — occur. They introduce two self-stabilizing protocols: an unbounded state space protocol that propagates auditable events across all nodes and a more efficient bounded state space protocol that achieves the same awareness without increasing resource demands. Their approach ensures that before a system can recover from an auditable event, all processes are aware of it, preventing unauthorized restorations and reinforcing accountability.

Access control objects regulate access to resources by defining policies for granting or denying permissions. Some of the most common access control mechanisms include Access Control Lists (ACLs) \cite{sandhu1994access}, which specify allowed or denied actions for individual users or processes, and Capability-Based Access Control (CBAC) \cite{levy2014capability}, where access is granted through transferable tokens rather than predefined rules. Role-Based Access Control (RBAC) \cite{sandhu1998role} simplifies management by assigning permissions to roles instead of individuals, making it widely used in enterprise environments, while Attribute-Based Access Control (ABAC) \cite{hu2015attribute} extends RBAC by dynamically evaluating attributes such as user location, device type, and time of access. AllowLists and DenyLists~\cite{FreyGR23} further control access by explicitly specifying permitted or blocked users or processes, often used in security filters and authentication systems.

An area related to auditability is \textit{data provenance} (also called \textit{lineage}) in databases \cite{becker1988auditing,glavic2021data}, which involves tracking the origin and transformations of data, focusing on its lineage across different processes or systems. Provenance is concerned with the integrity, quality, and reproducibility of data as it moves through various stages, such as aggregation or computation. While both auditability and provenance track interactions with data, auditability is more focused on recording who performed an operation, primarily for access control and security, while provenance aims to provide transparency about the history and transformations of the data for purposes such as verification and accountability in data analysis.

\subsection*{Summary of Our Contributions}

\begin{itemize}
    \item We extend auditability from single-writer to multi-writer registers and 
    characterize the necessary and sufficient consensus number needed for their implementation.
    \item We introduce and implement auditable LL/SC objects, 
    demonstrating auditability can be provided beyond read / write operations.
    \item We construct an anti-flickering \DenyList object from auditable registers, 
    illustrating the connection between auditability and access control mechanisms.

\end{itemize}

The remainder of this paper is organized as follows. Section~\ref{sec:definitions} formally defines auditable shared objects and presents our model. Sections~\ref{sec:mwa-cn-lower-bound} and~\ref{sec:fast_mwmr_auditable} prove the necessary and sufficient conditions for auditable multi-writer registers, respectively. Section~\ref{sec:LLSC} describes how to implement auditable LL/SC objects. Section~\ref{label:CN_of_DL} presents the algorithm to implement an anti-flickering \DenyList object with auditable registers. 
Finally, Section~\ref{sec:conclusions} concludes with open questions and future directions.


\section{Definitions}
\label{sec:definitions}

We use a standard model, in which a set of processes $p_1, \dots, p_n$, 
communicate through a shared memory consisting of \emph{base objects}.
The base objects are accessed with \emph{primitive operations}.
In addition to atomic registers, 
our implementations use $k$-sliding registers, 
whose consensus number is exactly $k$~\cite{MostefaouiPR18}.
Specifically, 
a $k$-\emph{sliding register}~\cite{MostefaouiPR18} stores the sequence 
of the last $k$ values  written to it (or the last $x$ values 
when only  $x < k$ values have been written). 
A $write$ with input $v$ appends $v$ at the end of the sequence and 
removes the first one if the sequence is already of size $k$. 
A $read$ operation returns the current sequence. 
A standard read / write register is a $1$-sliding register.

An \emph{implementation} of a (high-level) object $T$ specifies 
a program for each process and each operation of the object $T$;
when receiving an \emph{invocation} of an operation,
the process takes \emph{steps} according to this program. 
Each step by a process consists of some local computation,
followed by a single primitive operation on a base object.
The process may change its local state after a step, and it
may return a \emph{response} to the operation of the high-level object.

In order not to confuse operations performed on the implementation of the high-level object $T$ 
and primitives applied to base objects, the former are denoted with 
capital letters and the later in normal font.

A \emph{configuration} $C$ specifies the state of every process and of every base object.
An \emph{execution} $\alpha$ is an alternating sequence of configurations and events, 
starting with an \emph{initial configuration}; it can be finite or infinite.
An operation \emph{completes} in an execution $\alpha$ if $\alpha$ 
includes both the invocation and response of the operation;
if $\alpha$ includes the invocation of an operation, 
but no matching response, then the operation is \emph{pending}.
An operation $op$ \emph{precedes} another operation $op'$ in $\alpha$
if the response of $op$ appears before the invocation of $op'$ in $\alpha$.

A \emph{history} $H$ is a sequence of invocation and response events;
no two events occur at the same time.
The notions of \emph{complete}, \emph{pending} and \emph{preceding} operations
extend naturally to histories. 

The standard correctness condition for concurrent implementations 
is \emph{linearizability}~\cite{HerlihyWing90}: intuitively,
it requires that each operation appears to take place instantaneously 
at some point between its invocation and its response.
Formally:

\begin{definition}
Let \m{A} be an implementation of an object $T$.
An execution $\alpha$ of \m{A} is \emph{linearizable} if there is 
a sequential execution $\lambda(\alpha)$ (a \emph{linearization} 
of the operations on $T$ in $\alpha$) 
such that:
\begin{itemize}
    \item $\lambda(\alpha)$ contains all complete operations in $\alpha$, 
    and a (possibly empty) subset of the pending operations in $\alpha$ 
    (completed with response events),
    \item 
    If an operation $op$ precedes an operation $op'$ in $\alpha$, 
    then $op$ appears before $op'$ in $\lambda(\alpha)$, and
    \item $\lambda(\alpha)$ respects the sequential specification of the high-level object $T$.
\end{itemize}
\m{A} is \emph{linearizable} if all its executions are linearizable. 
\end{definition}

An implementation is 
\emph{wait-free} if, whenever there is 
a pending operation by process $p$, 
this operation returns in a finite number of steps by $p$.


An \emph{auditable register} supports, 
in addition to the standard $\rd{}$ and $\wrt{}$ operations, 
also an $\adt{}$ operation that reports which values were read by each process~\cite{AttiyaPMPR23}. 
An \adt{} has no parameters and it returns a set of pairs, 
$(j,v)$, where $j$ is a process id, and $v$ is a value of the register.
A pair $(j,v)$ indicates that process $p_j$ has read the value $v$.
The sequential specification of an auditable register enforces, 
in addition to the usual specification of \rd{} and \wrt{} operations,
that a pair appears in the set returned by an \adt{} operation if and only if 
it corresponds to a preceding \rd{} operation.

We implement a \emph{load-linked / store-conditional} (LL/SC) variable, 
supporting the following operations:
$\LL(x)$ returns the value stored in $x$, 
and $\SC(x, \mathit{new})$ writes the value $\mathit{new}$ to $x$,
if it was not written since the last $\LL(x)$ performed by the process; otherwise, $x$ is not modified. 
$\SC$ returns \textit{true} if it writes successfully, 
and \textit{false} otherwise.
An \emph{auditable LL/SC variable adds 
an \adt{} operation, whose sequential specification} returns a set of process-value pairs, 
corresponding to preceding \LL{} operations.

 
\section{\texorpdfstring{Consensus number of $n$-writer, $m$-reader auditable register $\geq m+n$}{Consensus number of n-writer, m-reader auditable register is at least m+n}}
\label{sec:mwa-cn-lower-bound}

An \adReg{m}{n}
can be written
by $n$ processes, read by $m$ processes, and be audited by all processes. 
An audit operation returns a set of pairs $(p,v)$ where $p$ is
a process id, and $v$ is a value. 
This set holds are the values returned by the read operations that precede the audit.

\begin{theorem}
  \label{thm:mwar_cn}
  For every pair of integers $m, n > 0$, 
  there is an $(m+n)$-process consensus algorithm using \adRegs{m}{n}.
\end{theorem}

\begin{proof}
The proof is by induction on $\ell = n+m$. 
The base case, $\ell=2$, 
is proved in~\cite[Proposition 19]{AttiyaPMPR23}.

For the induction step, assume the lemma holds for all pairs of values $n', m'$,
such that $n' + m' = \ell \geq 2$,
and we prove it for $\ell+1 > 2$. 
Pick $m>0, n>0$ such that $n+m = \ell+1$;
note that $m, n \leq \ell$.
To solve consensus among $\ell+1$ processes $p_1,\ldots, p_{\ell+1}$, 
we partition the processes into two sets:
$R$ (the \emph{readers}) of size $m$,
and $W$ (the \emph{writers}) of size $n$. 
Since $m>0$ and $n > 0$, both sets $R$ and $W$ are non-empty.

Each process $p_i, 1 \leq i \leq \ell+1 $ has a standard
single-writer multi-reader register $S_i$. 
By the induction hypothesis, since $n \leq \ell$,
the $n$ writers can agree on one of their proposals 
with an $n$-process consensus that uses \adRegs{1}{n-1}. 
Similarly, the $m$ readers can agree on one of their proposal 
using \adRegs{m-1}{1}. 
Each reader and writer process $p_i$ writes the value agreed upon in
its register $S_i$. 

We use now one \adReg{m}{n}, $AR$, whose initial value is $\bot$.
Writers and readers access $AR$ to select 
one of the two consensus values.
Each writer $\in W$  writes  $\top$ to $AR$ and 
each reader $\in R$ performs a read operation on $AR$. 
Writers and readers then audit $AR$.
If the set returned is empty or contains no pair
      $(p_r,\bot)$ where $r \in R$, then the first operation performed
      on $AR$ is a write and the writers win.
      Then, the decision for consensus can be read from at  least one of the registers of the
      writers. 
Otherwise, the  audit operation returns a set containing a
      pair $(p_r, \bot)$ with $r \in R$. 
      In that case, the first operation performed
      on $AR$ is a read operation and the readers win. As in the
      previous case, the decision for consensus can be found by
      reading the readers' registers.\qedhere
\end{proof}

\section{\texorpdfstring{Implementing an $n$-Writer $m$-Reader Auditable Register Using $(n+m)$-Sliding Registers}{n-Writer m-Reader Auditable Register Using (n+m)-Sliding Registers}}
\label{sec:fast_mwmr_auditable}

We present a wait-free and linearizable implementation of an 
$n$-writer $m$-reader auditable  register; 
it can support any number of auditors. 
The implementation uses $(m+n)$-sliding registers. 
Since the consensus number of \adReg{m}{n}s is $m+n$ 
(Theorem~\ref{thm:mwar_cn}), objects with consensus number $\geq m+n$,
like $(m+n)$-sliding registers, are required.
(Objects with consensus number $\infty$, like \emph{compare\&swap}, 
can also be used~\cite{AttiyaAFMRT2025}, 
but our goal is to use objects with the minimal consensus number.)

\begin{theorem}
\label{thm:mrmw_bounded}
There is a wait-free linearizable implementation of an $m$-reader, 
$n$-writer auditable register from $(m+n)$-sliding registers.
The step complexity of each operation is in $O(m+n)$. 
\end{theorem}

\subsection{The Algorithm}

An implementation of an auditable register has to keep track of the
latest value written to the register as well as, for each written value, its set of
readers. This can be easily achieved using a single sliding register
$\mathit{SLR}$, provided that its window is unbounded. Such a register
hence stores the complete sequence of values written to it, ordered
by the oldest first but has infinite consensus number. To perform
a \wrt($v$), a writer simply writes $v$ to
$\mathit{SLR}$. 
For auditing purpose, a reader $p_i$ first writes its identifier $i$ to
the sliding register, before reading it. The value returned by this \rd{}
 is then the nearest non-identifier value in the sequence read
from $\mathit{SLR}$ previous to the identifier $i$ written by $p_i$.
The reader set of each value  can easily be inferred
from the sequence  stored in $\mathit{SLR}$. For a value $v$ in the
sequence, its reader set is the set of processes whose identifiers
follows $v$, and are before the first non-identifier value that
succeeds $v$, if any.  

\begin{algorithm*}
\small
  \caption{Multi-writer auditable register: \rd{} and \wrt{} and \adt{}}
  \label{alg:mw_auditable}
  \begin{algorithmic}[1]
    \State\textbf{shared variables}
    \State\hspace{\algorithmicindent}$\mathit{M}$: a max-register storing a triple $(widx,ridx,auditset)$ ordered by their first field
    \State\hspace{\algorithmicindent} $~~~ widx$ initially $0$ 
        \Comment{index of sliding registers}
    \State\hspace{\algorithmicindent} $~~~ ridx$ array of size $m$, initially $[-1,\ldots,-1]$ 
    \Statex \Comment{$ridx[i]$ highest index of sliding register in which a \rd{} by  $p_i$ is recorded}
    \State\hspace{\algorithmicindent} $~~~ auditset$ set of pair (process,value), initially $\emptyset$
    \State\hspace{\algorithmicindent}$SLR[-1,0,..+\infty]$: unbounded array of $(m+n)$-sliding registers,\label{l:SLR}
    \Statex\hspace{\algorithmicindent}\hspace{\algorithmicindent} 
        initially, $SLR[-1] = ((\mathsf{w},j_0,v_0,\emptyset))$, $SLR[\ell] = ()$ for all $\ell \geq 0$
    \State\hspace{\algorithmicindent}$H[1..m]$: array of SWMR register, one per reader, initially $[-1,\ldots,-1]$\label{l:H}

    \State\textbf{local variables: reader}    
    \State\hspace{\algorithmicindent}$lsr \gets -1$; $lval \gets \bot$  \Comment{index of last sliding register read and last value read}\label{l:local_seq_r} 

    \Function{read}{$~$} \Comment{code for reader process $p_i, i \in \{1,\ldots,m\}$}
    \If{$lsr \geq 0$} $window \gets SLR[lsr].\mathsf{read}()$\label{l:rd_last_wd} \Comment{check for new \wrt{} since last \rd{}}
    \If{there is no $w$-tuple $(\mathsf{w},\_,\_,\_)$ in $\mathit{window}$} \Return{$lval$}
    \Comment{no new \wrt{}}\label{l:no_new_write}\EndIf
    $(widx,ridx,auditset) \gets M.\mathsf{read}()$\label{l:new_write} \Comment{new \wrt{}, check if it needs help to complete}
    \If{$widx = lsr$}
    \State \textbf{for each} $j \in \Call{readers}{window}$  \textbf{do} 
    \State\hspace{\algorithmicindent} $ridx[j] \gets lsr$; $auditset \gets auditset \cup \{(j,lval)\}$\label{l:rd_update_ridx1}\label{l:rd_help_A1}
    \State $M.\mathsf{writeMax}(lsr+1,ridx,auditset)$\label{l:rd_help_M1}
    \EndIf
    \EndIf
    \Repeat{}\label{l:rd_start_repeat}
    \State $(widx,ridx,auditset)  \gets M.\mathsf{read}()$ \label{l:rd_M}
    \If{$ridx[i] > lsr$} \Comment{found help} 
    \State $lsr \gets ridx[i]$; $lval \gets$ \Call{getValue}{$lsr-1$} \Return $lval$ \label{l:rd_return_help}
    \EndIf
    \State $lsr \gets widx$; $H[i].\mathsf{write}(lsr)$ \label{l:rd_annouce_att} \Comment{catch up with writers, announce attempt}
    \State $SLR[lsr].\mathsf{write}(i)$; $window \gets SLR[lsr].\mathsf{read}()$; $lval \gets $\Call{getValue}{$lsr-1$}\label{l:rd_do_att}\label{l:rd_val}
    \If{$\exists (\mathsf{w},\_,\_,\_) \in window$} \Comment{help corresponding \wrt{} to complete}\label{l:rd_help_write}
    \State \textbf{for each} $j \in \Call{readers}{window}$  \textbf{do} 
    \State\hspace{\algorithmicindent} $ridx[j] \gets lsr$; $auditset \gets auditset \cup \{(j,lval)\}$\label{l:rd_update_ridx2}\label{l:rd_help_A2}
    \State $M.\mathsf{writeMax}(lsr+1,ridx,auditset)$ \label{l:rd_help_M2}
    \EndIf
    \Until{$i \in \Call{readers}{window}$}\label{l:rd_att_success}
    \State \Return{$lval$} \label{l:rd_return_no_help}
    \EndFunction
     
    \Function{write}{$v$} \Comment{code for writer $p_j$, $j \in \{1,\ldots,n\}$}
    \State $widx,ridx,auditset  \gets M.\mathsf{read}()$;  $to\_help  \gets \emptyset$ \label{l:w_read_M}
    \ForAll{$i \in \{1,\ldots,m\}$} $aidx_j \gets H[j].\mathsf{read}()$ \Comment{help ongoing read operations}\label{l:w_att}
    \If{$ridx[j] <  aidx_j$} \Comment{reader $p_j$ may need help} \label{l:w_help_needed?}
    \State $window \gets SLR[aidx_j].\mathsf{read}()$  \label{l:w_help_wd}
    \State \textbf{if} $j \notin \Call{readers}{window}$ \textbf{then} $to\_help \gets to\_help \cup  \{j\}$ \label{l:w_att_fail}
    \EndIf
    \EndFor
    \State $SLR[widx].\mathsf{write}(\mathsf{w},j,v,to\_help)$;\Comment{announce new write}\label{l:w_write_SLR}
    \State $window \gets SLR[widx].\mathsf{read}()$; $val \gets \Call{getvalue}{widx-1}$
    \State \textbf{for each} $j \in \Call{readers}{window}$  \textbf{do} 
    $ridx[j] \gets widx$; $auditset \gets auditset \cup \{(j,val)\}$ \label{l:w_update_ridx}\label{l:w_A}
    \State $M.\mathsf{writeMax}(widx+1,ridx,auditset)$; \Return\label{l:w_write_M}
    \EndFunction
    \Function{audit}{$~$}
    \State $widx,\_,auditset  \gets M.\mathsf{read}()$ \label{l:adt_read_M} \Comment{readers of val with seq. num $\leq widx-2$}
    \State $window \gets SLR[widx].read()$; $val \gets \Call{getValue}{widx-1}$ \label{l:adt_val}
    \State $auditset \gets auditset \cup \{(j,val) : j \in \Call{readers}{window}\}$ 
                \Comment{readers of val with seq. num $widx-1$}\label{l:adt_current_val}
    \State \Return{$auditset$}\label{l:adt_return} 
    \EndFunction
    \algstore{alg:fast_mwa_break}
  \end{algorithmic}
\end{algorithm*}

\begin{algorithm*}
\small
  \caption{Multi-writer auditable register: auxiliary functions \textsc{getValue} and  \textsc{readers}}
  \label{alg:mw_auditable_aux}
  \begin{algorithmic}[1]
    \algrestore{alg:fast_mwa_break}
   
    \Function{getValue}{$sn$}
    \State $window \gets SLR[sn].\mathsf{read}()$; let $(\mathsf{w},id,val,\_)$ be the first $w$-tuple $(\mathsf{w},\_,\_,\_)$ in $window$
    \label{l:getValue}
    \State \Return{$val$}
    \EndFunction

        \Function{readers}{$window$}
    \State $readers \gets \{j: ~\mbox{there is no $w$-tuple $(\mathsf{w},\_,\_,\_)$ preceding $j$ in $window$}\}$
    \label{l:read_stop_1}
    \If{$\exists (\mathsf{w},\_,\_,\_) \in window$}
    \State let $(\mathsf{w},id,\_,h)$ be the 1st $(\mathsf{w},\_,\_,\_)$ in window; $readers \gets readers \cup h$
    \label{l:read_stop_2}
    \EndIf
    \State \Return{$readers$}
    \EndFunction
  \end{algorithmic}
\end{algorithm*}

Our implementation (Algorithms~\ref{alg:mw_auditable}
and \ref{alg:mw_auditable_aux}) is based on this simple
idea, but instead of a single sliding register with an
unbounded window, we use an unbounded array $SLR[-1,\ldots]$ of
sliding registers, each with a bounded window of size $m+n$.
In order not to confuse identifiers and written values, \wrt{}
operations insert into the sliding registers $w$-\emph{tuples} of the form $(\mathsf{w},j,v,h)$, where
$j$ the identifier of the writer, $v$ is the input value of the \wrt{} operation, and
$h$ a \emph{helping set} of readers' identifiers
 (whose role will be explained later).

Each sliding  register $\mathit{SLR}[x]$ contains initially the
empty sequence $()$, except the first $\mathit{SLR}[-1]$ whose
sequence contains the $w$-tuple $(\mathsf{w},j_0,v_0,\emptyset)$ where
$v_0$ is the initial value of the auditable register and $j_0$ an
arbitrary writer's identifier. At any point in the execution, the
current value $v$ of the auditable register is found in the sliding
register with highest index $x$ whose sequence contains a
$w$-tuple. Specifically, $v$ is the value contained in the \emph{first}
$w$-tuple in the sequence stored in $\emph{SLR}[x]$. Similarly to the basic
implementation sketched above, readers of this value
(if any) are the processes whose identifiers appear \emph{before} any
$w$-tuple in the \emph{next} sliding register $\emph{SLR}[x+1]$.

Therefore, reading or writing the auditable register
involves finding the \emph{valid} sliding register (that is, one that
does not contain a $w$-tuple)  with lowest index. If this is not done with
care, progress of some operation may be lost. For example, the same
writer may be always the first to write in each sliding register
(i.e., each non empty sequence in any sliding register starts with a
$w$-tuple posted by this writer), thus preventing \rd{} operations from
completing  or other writers from changing the value of the auditable
register. We also have to make sure that each writer or reader writes
at most once to each sliding register. Otherwise, the auditing may become
inaccurate (as readers identifiers may be removed from some sequence),
or the current value of the auditable register may be lost.

To solve theses challenges, the implementation combines the following
ideas:
\begin{itemize}
\item First, we observe that a \wrt{} operation can terminate 
after writing in a given sliding register $\mathit{SLR}[x]$, even if its
  corresponding $w$-tuple is not the first in the sequence held in
  $\mathit{SLR}[x]$, provided that that it is concurrent with the
  \wrt{} operation that writes first its $w$-tuple into $SLR[x]$. 
  Indeed, in that case, 
  the \wrt{} operations that are late to post their $w$-tuple may be linearized
  immediately \emph{before} the \wrt{} whose $w$-tuple is first.
  Accordingly, the value, denoted $v_x$, in the first $w$-tuple in the
  sequence stored in $SLR[x]$ is said to be \emph{visible}. The values in
  the other $w$-tuple are never returned by any \rd{} operation, and
  are thus \emph{invisible}. 
  
\item Second, we use a \emph{max register}\footnote{
    Recall that a max register retains the largest value written to it. 
    Wait-free and linearizable max registers can be implemented from 
    atomic read/write registers with linear step complexity
    \cite{JACM_monotone}.}  $M$  to store the  current
  smallest index $widx$ of the still valid sliding registers.
  A \wrt{} thus starts by retrieving this index from $M$, 
  writes its $w$-tuple in $SLR[widx]$, and finally updates $M$ with the
  new index  $widx+1$. This ensures that, if several  \wrt{} operations
  post $w$-tuples in the same sliding register $SLR[x]$, they are
  concurrent.  
  
  In fact,  in addition, $M$ stores an
  $m$-vector $ridx$ of indexes of sliding registers (which is part of
  the helping mechanism described next) and, for convenience, the
  $auditset$ of the values written whose set of readers is already completely determined
  (the set of readers of each visible value $v_x$,  for $x < widx$).
  Triples $(widx, ridx, auditset)$ are ordered by their first field.
 An \adt{} operation therefore reads $M$, and for $M.widx=x$,  
 gets the \emph{definitive} readers of values $v_{x'}, x' \leq x-2$ from $M.auditset$, to which it adds the possibly \emph{non-definitive} readers of $v_{x-1}$ by reading $SLR[x]$. 
\item Finally, a helping mechanism is used to make sure that \rd{}
  operations are wait-free. The set of readers of the visible value $v_x$ (that appears in
  the first $w$-tuple stored in $SLR[x]$) are the processes whose
  identifiers are before any $w$-tuple in  $SLR[x+1]$ \emph{or} are
  in the helping set $h$ of the first $w$-tuple of $SLR[x+1]$. An
  additional array $H$ indicates, for each reader, the index of the
  latest sliding register in which it attempted to write its identifier
  before any $w$-tuple.  The array $ridx$ (stored, as seen above, in
  the max register $M$), 
  records for each reader the highest index of a sliding register in
  which it has received help (that is, the identifier of the reader is
  included in the helping set of the first $w$-tuple of that sliding
  register). A writer $p_j$ hence determines if a given  reader $p_i$ needs
  help by comparing $H[i]$ and $M.ridx[i]$. If $M.ridx[i] < H[i]$,
  $p_i$ has an ongoing \rd{} operation, and $i$ is therefore added to the
  helping set of the $w$-tuple of the writer. Similarly, reader $p_i$
  discovers if it has received help by comparing $M.ridx[i]$ with the
  index of the latest sliding register in which it writes its
  identifier. 
\end{itemize}

Hence, besides $w$-tuples, a given sliding register  $SLR[x]$ may also include identifiers $i$ of readers, 
with the following meaning: 
\begin{enumerate}
\item If $i$ appears \emph{before} the first $w$-tuple in $SLR[x]$,
then $p_i$'s \rd{} returns $v_{x-1}$, 
which is the value stored in the first $w$-tuple in $SLR[x-1]$.

\item If $i$ appears in the helping set of the first $w$-tuple in $SLR[x]$, 
then $p_i$'s \rd{} also returns $v_{x-1}$ as in the previous case.

\item If $i$ appears \emph{after} the first $w$-tuple, 
    then it is too late and fails to read $v_{x-1}$.

\end{enumerate}

The reader processes satisfying Cases (1) and (2) are said to be \emph{recorded}
in $SLR[x]$.
The function {\sc readers} in Algorithm~\ref{alg:mw_auditable_aux}
returns the ids of readers recorded in a sliding register (which is
used, in particular, to update the audit set). 
The processes satisfying Case (2) are called \emph{helped} in $SLR[x]$.

Regarding Case (3), we prove that after at most two failed attempts,
a reader $p_i$ receives help, or succeeds in appearing before  any $w$-tuple in a sliding register. The algorithm also ensures that a \rd{} operation is helped at most once. Before an attempt to write its id into $SLR[x]$, reader $p_i$  knows, 
by checking if $ridx[i]  \geq H[i]$, if it has already received help. 
If this is the case, it directly returns the corresponding value 
without updating $H[i]$.

\subsection{Proof of Correctness} 

In this section, we prove Theorem~\ref{thm:mrmw_bounded},
starting with basic properties.
We show that each operation returns within $O(m+n)$ of its own steps,
and then explain how to construct a sequential 
execution $\lambda$ that contains all completed operations in an execution $\alpha$, 
and some pending operations.
Lemma~\ref{lem:rt_order} shows that 
$\lambda$ preserves the real-time order between operations in $\alpha$. 
By Lemma~\ref{lem:read_last}, Lemma~\ref{lem:adt_complete}, and   Lemma~\ref{lem:adt_accurate}, $\lambda$ is a sequential
execution of a $m$-reader, $n$-writer auditable register.  

To proceed with the detailed proofs, 
fix a finite execution $\alpha$ of the algorithm.
We start with basic properties. 
Lemma~\ref{lem:write_once} 
shows that that each sliding register retains the complete history 
of the $\mathsf{write}$s applied to it. 

\begin{restatable}{lemma}{writeonce}
\label{lem:write_once}
Each process applies at most one $\mathsf{write}$ to $SLR[k]$, 
 for any $k \geq 0$, 
\end{restatable}
\begin{proof}
  \adt{} operations only read shared objects.
  A \wrt{} operation $op$ by $p_j$ applies a single $\mathsf{write}$ to some sliding register $SLR[x]$ (line~\ref{l:w_write_SLR}) after reading $x$ from  $M.widx$ (line~\ref{l:w_read_M}).
  After $op$ completes (line~\ref{l:w_write_M}), $M.widx > x$ . 
  
  A \rd{} operation $op$ by process $p_i$ makes several attempts to place $i$ before every $w$-tuple in some sliding register. 
  Each attempt (an iteration of the repeat loop) applies at most one  $\mathsf{write}$ to a distinct sliding register (line~\ref{l:rd_do_att}). Process $p_i$ starts by reading $x$ from $M.widx$ before  writing $i$ to $SLR[x]$, and, if the attempt is not successful (the sequence  read from $SLR[x]$ contains an invalidating  $w$-tuple preceding $i$), the code makes sure that $M.widx > x$ (line~\ref{l:rd_help_M2}) 
  before starting a new attempt. 

  It remains to ensure that the first attempt of $op$  does not write to some sliding  register $SLR[x']$ in which a previous operation $op'$ by $p_i$ has already written. 
  The local persistent variable $lsr$ keeps track of the highest index of a sliding register in which $p_i$ has written, and hence,
  $lsr \geq x'$ when $op$ starts.  
  In $op$, $p_i$ enters the repeat loop only if $SLR[lsr]$ contains an invalidating tuple $(\mathsf{w},\_,\_,\_)$ (line~\ref{l:no_new_write}), and if this the case, makes sure that $M.widx > lsr \geq x'$ (line~\ref{l:rd_help_M1}) before the loop.
\end{proof}

The max register $M$ stores a triple $(widx,ridx,auditset)$, 
where $ridx$ is a $m$-vector and $audiset$ a set of (process,value) pairs. 
Triples are ordered in increasing order of their first element. 
At any point in the execution, the triple stored in $M$ is thus 
a triple with the largest first member written to $M$.
We also partially order $m$-vectors as follows: $ridx \leq ridx'$ if and only $\forall i \in \{1,\ldots,m\}, ridx[i] \leq ridx'[i]$. 

\begin{proposition}
\label{prop:widx}
  The successive values of $M.widx$ are $0,1,2,\ldots$.
\end{proposition}

The function \textsc{readers} extracts from a sequence read from some 
sliding register $SLR[x]$ a set of processes. 
This set does not change once a $w$-tuple has been written to $SLR[x]$. 

\begin{proposition}
\label{lem:same_readers}
If $win$ and $win'$ are two sequences read from $SLR[x]$ after  $SLR[x].\mathsf{write}(\mathsf{w},\_,\_,\_)$ has been applied,
then $\textsc{readers}(win) = \textsc{readers}(win')$.
\end{proposition}

\begin{proof}
Let $(\mathsf{w},id,val,h)$ be the first $w$-tuple written to $SLR[x]$. 
By Lemma~\ref{lem:write_once}, $(\mathsf{w},id,val,h)$ is the first $w$-tuple in both $win$ and $win'$, 
and the identifier $j$ of a reader precedes $(\mathsf{w},id,val,h)$ in $win$ if and only if $j$ is also before this $w$-tuple in $win'$.  
By the code, $\textsc{readers}(win) = \textsc{readers}(win')$. 
\end{proof}

The following lemma is proved by induction on $k$.

\begin{restatable}{lemma}{sameridx}
\label{lem:same_ridx}
Suppose $\mathsf{writeMax}(k, iv,\_)$ and $\mathsf{writeMax}(k', iv',\_)$ are applied to $M$. 
(1) if $k = k'$ then $iv = iv'$ and, (2) if $k < k'$ then $iv \leq iv'$. 
\end{restatable}
\begin{proof}
  The proof is by induction. 
  For the base case ($k=0$), the lemma is immediate since no $\mathsf{writeMax}(0, \_, \_)$ is applied to $M$. 
  Indeed, by the code, $p$ applies $M.\mathsf{writeMax}(k,\_, \_)$ (line~\ref{l:rd_help_M1}, line~\ref{l:rd_help_M2}, or line~\ref{l:w_write_M}) after reading $k-1$ from $M.widx$ (line~\ref{l:new_write}, line~\ref{l:rd_M}, or line~\ref{l:w_read_M}, respectively). 
  Proposition~\ref{prop:widx} implies that $M.widx \geq 0$, 
  and therefore, there is no $M.\mathsf{writeMax}(0,\_,\_)$.

  For the induction step, assume that $k \geq 0$, 
  and assume that the lemma holds for any $\ell \leq k$.
  By the induction hypothesis on (1),
  $M.ridx$ holds a fixed vector $ridx_\ell$, as long as $M.widx = \ell$. 
  Consider the first $\mathsf{writeMax}(k+1,iv,\_)$, by some process $p$,   
  that changes $M.widx$ from $k$ to $k+1$. 
  Before this step, $p$ reads $(k,ridx_{k},\_)$ from $M$ (line~\ref{l:rd_help_M1}, line~\ref{l:rd_help_M2}, or line~\ref{l:w_write_M}), 
  and $SLR[k]$ already contains a $w$-tuple 
  (as seen by $p$ after reading $SLR[k]$, if $p$ is performing a \rd{} (line~\ref{l:rd_last_wd} or line~\ref{l:rd_do_att}), or the tuple is written by $p$ if it is performing a \wrt{} (line~\ref{l:w_write_SLR})). 
  Then, $iv$ depends solely on $ridx_k$ and the current set $readers$ in  $SLR[k]$ that is read by $p$.  
  More precisely, this set is extracted from a sequence  
  that contains a $w$-tuple. For each $i \in \{1,\ldots,m\}$, $iv[i] =ridx_k[i]$  if $i$ is not in the reader set, and $iv[i] = k$ otherwise.  
  Note that $iv \geq ridx_k$, implying (2). 

  Consider now another $\mathsf{writeMax}(k+1,iv',\_)$. 
  As for the first such step, $iv'$ is inferred from $iv_k$ and a reader set $readers'$ extracted from a value read from $SLR[k]$ after a $w$-tuple has been written to $SLR[k]$. 
  By Proposition~\ref{lem:same_readers}, $readers = readers'$ 
  and hence $iv' =iv$. 
\end{proof}

By Proposition~\ref{prop:widx}, the successive values of $M.widx$ are $0,1,\ldots$
By Lemma~\ref{lem:same_ridx}(1), $M.ridx$ remains the same  while $M.widx$  does not change. 
Lemma~\ref{lem:same_ridx}(2) implies that the successive vectors in $M.ridx$ are ordered and form an increasing sequence. 

Let  $\mathit{WIDX}$ denote the  highest index of a sliding register 
to which a $w$-tuple has been written. 
That is, at the end of a finite execution $\alpha$, 
$\mathit{WIDX} = k$ if and only a $w$-tuple was written in $SLR[k]$, 
and no $w$-tuple was written in $SLR[k']$, for any $k' > k$. 
When $\mathit{WIDX}$ is changed to $k+1$ as a result of some process $p$ applying $\mathsf{write}(\mathsf{w},\_,\_,\_)$ to $SLR[k+1]$, $M.widx \geq k+1$. 
Indeed, before writing a $w$-tuple to $SLR[k+1]$, 
$p$ has read $k$ from $M.widx$ in line~\ref{l:w_read_M}. 
Observe also that when $M.widx$ is changed from $k$ to $k+1$ 
by some process $p$, $M.widx \geq k$. 
Indeed, before applying $M.\mathsf{writeMax}(k+1,\_, \_)$, 
$p$ reads a sequence from $SLR[k]$ that contains a $w$-tuple (line~\ref{l:rd_last_wd} or line~\ref{l:rd_val}, $p$ is performing a \rd{}), 
or has written a $w$-tuple to $SLR[k]$ 
(line~\ref{l:w_write_SLR}, $p$ is performing a \wrt{}). 
This implies that $M.widx$ is always in $\{\mathit{WIDX},\mathit{WIDX}-1\}$,
which shows:

\begin{proposition}
\label{prop:lower_invalidated}
For every $x \geq 0$, 
if $SLR[x]$ is accessed then $SLR[x-1]$ contains a $w$-tuple. 
\end{proposition}


Combining these observations with Lemma~\ref{lem:same_ridx}(2), we have:

\begin{lemma}
\label{lem:phase}
For some $k \geq 0$, the finite execution $\alpha$ can be written as 
either $D_0\rho_0E_0\mu_1D_1\rho_1E_1\ldots{}\mu_kD_k$ 
or $D_0\rho_0E_0\mu_1D_1\rho_1E_1\ldots{}\mu_kD_k\rho_kE_k$, where:
  \begin{itemize}
  \item $\rho_\ell$ is the first  step that writes a $w$-tuple to $SLR[\ell]$ (applied by a writer, line~\ref{l:w_write_SLR}), $\mu_\ell$ is the step that changes $M.widx$ from $\ell-1$ to $\ell$ (applied within a \rd{}, line~\ref{l:rd_help_M2} or line~\ref{l:rd_help_M1}, or a \wrt{}, line~\ref{l:w_write_M}). 
  \item in any configuration in $D_\ell$, $M.widx = \ell = \mathit{WIDX}+1$, and in any configuration in $E_\ell$, $M.widx = \ell = \mathit{WIDX}$
  \end{itemize}
\end{lemma}

Therefore, 
the first step  that writes a $w$-tuple to $SLR[x]$ is preceded 
by steps $\rho_0,\ldots,\rho_{x-1}$ that write $w$-tuples to
$SLR[0],\ldots,SLR[x-1]$, and they are never deleted 
(by Lemma~\ref{lem:write_once}).

Let $x_0$ be the value of the local variable $lsr$ when \rd{} 
operation $op$ by process $p_i$ starts.

\begin{proposition}
  \label{prop:att_x_increase}
  If $op$ enters the repeat loop (line~\ref{l:no_new_write}), then $x_0 < x_1 < x_2 < \ldots$ where $x_k$, $k > 0$, is  the value read from $M.widx$ (line~\ref{l:rd_M})  at the beginning of the $k$th iteration.
\end{proposition}


If $op$ does not terminate in the $k$th iteration of the loop
(line~\ref{l:rd_return_help}), 
after reading $(x_k,\_,\_)$ from $M$,
then $x_k$ is written to $H[i]$ (line~\ref{l:rd_annouce_att}) 
before an attempt is made to place $i$ ahead of any $w$-tuple in $SLR[x_k]$. 
We further prove:

\begin{restatable}{lemma}{nohelp}
  \label{lemma:no_help}
  If $x_k$ is written to $H[i]$  
  then  $i \notin~ $\textsc{readers}$(win)$, for any sequence $win$ in any sliding register $SLR[x], x_0 < x < x_k$.
\end{restatable}
\begin{proof}
First, assume that $x = x_\ell$ for some $\ell, 0 < \ell < k$. 
As $op$ does not terminate after the $\ell$th iteration,
   (1) $i$ is written to $SLR[x_\ell]$ (in line~\ref{l:rd_do_att}) and (2) this writes happens after a $w$-tuple has already been written. Moreover, (3) the helping set of the first $w$-tuple does not contain $i$. 
If neither of (2) and (3) is true, then $i \in $\textsc{readers}$(win)$,
for $win$ read from $SLR[x_\ell]$, and $op$ terminates after iteration $\ell$. 
The lemma follows since the readers set does not change after a $w$-tuple has been written to $SLR[x_\ell]$ (Proposition~\ref{lem:same_readers}). 

Assume now that $x\neq x_\ell$ for every $\ell, 0 < \ell < k$. Note that $i$ is never written to $SLR[x]$. Suppose for contradiction that $i \in h_x$, where $h_x$ is  the helping set of the first $w$-tuple written to $SLR[x]$. Let $\ell, 0 < \ell \leq k$ such that $x_{\ell-1} < x < x_\ell$. As $x_\ell, x < x_\ell$ is read from $M.widx$, it follows from Lemma~\ref{prop:widx} that $M.widx$ is changed from $x$ to $x+1$. Let $p$ be the process that applies the corresponding $\mathsf{writeMax}$ steps. Before this step, $p$ has read a $win$ from $SLR[x]$ that contains a $w$-tuple (at line~\ref{l:rd_last_wd}, line~\ref{l:rd_do_att}, or line~\ref{l:w_write_SLR}; if there is no $w$-tuple in $win$, $M$ is not changed). Therefore, $i \in $\textsc{readers}$(win)$ as it is in the first helping set written to this sliding register, and hence (line~\ref{l:rd_update_ridx1}, line~\ref{l:rd_update_ridx2}, or line~\ref{l:w_update_ridx}), the $m$-vector $ridx_{x+1}$ written to $M$ together with $x+1$ is such that $ridx_{x+1}[i] = x$.

Consider now the $\ell$th iteration of the repeat loop by $p_i$. At the beginning of the iteration, $p_i$ reads $(x_\ell, ridx_{x_\ell}, \_)$ from $M$ with $ridx_{x_\ell}[i] \geq x$ (Lemma~\ref{lem:same_ridx}(2)). The  value of the local variable $lsr$ has not yet changed, and is therefore still equal to $x_{\ell-1} < x$. 
Thus, $op$ terminates (line~\ref{l:rd_return_help}) immediately after reading $M$, without writing $x_\ell$ to $H[i]$. If $\ell = k$, this is  a contradiction as $x_k$ is written to $H[i]$. If $\ell < k$, this also a contradiction as $op$ does not terminate after iteration $\ell$. 
\end{proof}

Finally, we prove 
that $op$ cannot find help in $SLR[x]$, 
for any $x, x_0 < x < x_1$ (Lemma~\ref{lem:no_help_before_calling}). 
\begin{restatable}{lemma}{iffridx}
\label{lem:readers_iff_ridx}
$M.ridx[i] \geq x$ if and only if 
$SLR[x]$ holds a sequence $win$ such that $i \in $\textsc{readers}$(win)$.
\end{restatable}

\begin{proof}
By Lemma~\ref{lem:same_ridx}(2), 
the successive values of $(M.widx,M.ridx)$ are $(0,ridx_0), (1,rdix_1),$ $\ldots, (k,ridx_k)$ with $ridx_0[i]  \leq ridx_1[i] \leq \ldots \leq ridx_k[i]$. 

$\impliedby$ It suffices to prove that $ridx_{x+1}[i] \geq x$. 
Consider the prefix $\alpha'$ of $\alpha$ that ends with a 
$M.\mathsf{write}(x+1,ridx_{x+1})$ operation by some process $p$ 
that changes $M.widx$ from $x$ to $x+1$.
Before this step, $p$ reads $win$ from $SLR[x]$ 
(at line~\ref{l:rd_do_att} or line~\ref{l:rd_last_wd} if $p$ is performing a \rd{}, 
or in line~\ref{l:w_write_SLR} if it is performing a \wrt{}). 
As $win$ contains a $w$-tuple, $i$ must precede any  $w$-tuple, 
or is in the helping step of the first $w$-tuple, e.g., 
$i$ is in the set returned by $\textsc{readers}(win)$ 
(line~\ref{l:rd_update_ridx1}, line~\ref{l:rd_update_ridx2}, or line~\ref{l:w_update_ridx}). 
Indeed, by Proposition~\ref{lem:same_readers}, $j$ cannot be added 
to the set of readers of $SLR[x]$ after a $w$-tuple has been written to it. 
It follows that that the $m$-vector written to $M$ is such that $ridx_{x+1}[i] = x$, 
as by the code, for $j \in $\textsc{readers}$(win)$, $ridx_{x+1}[i] = x$. 

$\implies$ 
Let $\ell \geq k$ be the smallest index such that $ridx_\ell[i] = x$.  
By the code (line~\ref{l:rd_update_ridx1}, line~\ref{l:rd_update_ridx2}, or line~\ref{l:w_update_ridx}), if $ridx[i]$ is set to $x$ by $p$, 
it is after $p$ reads $\emph{win}$ from $SLR[x]$ with $i\in \textsc{readers}(win)$. 
Therefore, $i$ appears before any $w$-tuple in $win$ or 
in the helping set of the first $w$-tuple in $win$, 
and the lemma follows by Proposition~\ref{lem:same_readers}. 
\end{proof}

\begin{restatable}{lemma}{callhelp}
\label{lem:no_help_before_calling}
For any sequence $win$ in any sliding register $SLR[x], x_0 < x < x_1$, $i \notin~ $\textsc{readers}$(win)$.
\end{restatable}
\begin{proof}
  Suppose that the lemma is not true, and let $x$, $x_0 < x < x_1$, be the smallest index for which $i \in $\textsc{readers}$(win)$, for some $win$ stored in $SLR[x]$. As $p_i$ does not write to $SLR[x]$,  $i$ must be in the helping set $h_x$ of the first $w$-tuple in $SLR[x]$.  Let $q$ be the process that writes this tuple, and let denote by $\sigma$ this step. By the code, before writing to $SLR[x]$ (line~\ref{l:w_write_SLR}), $q$ has read $r_i$ from $M.ridx[i]$ (line~\ref{l:w_read_M}), $a_i$ from $H[i]$ (line~\ref{l:w_att}) with $r_i < a_i$,  for $i \in h_x$ (line~\ref{l:w_help_needed?})). We show that this is not possible, i.e, $r_i \geq a_i$.

  If $op$ is $p_i$'s first \rd{} operation, $H[i] = -1$ when it is read by $q$, and hence $r_i \geq -1$, as $M.ridx[i]$ is increasing and its initial value is $-1$. Else, there exists a \rd{} operation $op'$ by $p_i$, preceding $op$, in which the local variable $lsr$ is set to $x_0$ in the repeat loop (line~\ref{l:rd_return_help} or line~\ref{l:rd_annouce_att}). Therefore, $SLR[x_0]$ holds a sequence $win$ with $i \in $\textsc{readers}$(win)$, and it follows from Lemma~\ref{lemma:no_help}  that $H[i]$ cannot be changed to a value $> x_0$ by $op'$ or a previous \rd{} operation by $p_i$. 
  Also, if there is a \rd{} operation $op''$ by $p_i$ between $op'$ and $op$, $op''$ returns immediately in line~\ref{l:no_new_write} after reading $SLR[x_0]$ as  $lsr = x_0$ when $op$ starts, and thus does not change $H[i]$. We conclude that $a_i \leq x_0$. 

  For $r_i$, note that when $M$ is read by $q$ (line~\ref{l:w_read_M}) before step $\sigma$, for every $x' > x_0$,  $i \notin \textsc{readers}(win)$ for any $win$ in $SLR[x']$, but $SLR[x_0]$ holds a sequence $win$ with $i\in$\textsc{readers}$(win)$. Hence, by Lemma~\ref{lem:readers_iff_ridx}, $q$ reads $M.ridx[i] =   x_0 = r_i$. Therefore, $a_i \leq x_0 = r_i$, which contradicts  $i \in h_x$. 
\end{proof}

\subsubsection{Termination and step complexity}
\label{sec:alg termination}

An \adt{} operation reads $M$ once and a single sliding register. 
A \wrt{} operation applies a single $\mathsf{writeMax}$ and a single $\mathsf{read}$ to $M$, reads $O(m)$ registers and writes $O(1)$ registers. 
Since there is a linearizable implementation of max registers for $(n+m)$ processes using registers with $O(m+n)$ step complexity per operation~\cite{JACM_monotone},
this implies that the step complexity of a \wrt{} or a \adt{} operation 
is in $O(m+n)$.

Lemma~\ref{lem:wait_free} 
shows that \rd{} operations are also wait-free, 
by proving that the repeat loop has at most 3 iterations. 
This is because a writer may place a $w$-tuple in at most one sliding register without detecting a concurrent \rd{} operation and helping it.

\begin{lemma}
\label{lem:wait_free}
The step complexity of a \rd{} operation is in $O(m+n)$.
\end{lemma}

\begin{proof}
Let $op$ be a \rd{} operation by process $p_i$.
We  prove that $op$ terminates within at most three iterations of the repeat loop (lines~\ref{l:rd_start_repeat}-\ref{l:rd_att_success}). 
Since each iteration applies a constant number of steps, 
the lemma follows. 

In iteration $1$, after reading  $x_1$ from $M.widx$ (line~\ref{l:rd_M}), $p_i$ writes $x_1$ to $H[i]$ (line~\ref{l:rd_annouce_att}). We denote by $\sigma_1$ this step. $\sigma_1$ occurs before $p_i$ reads $x_2$ from $M.widx$ in iteration $2$, which we denote by step $\sigma_2$. 
By Lemma~\ref{lem:phase}, $\emph{WIDX} \leq x_2$ immediately after this step (recall that $\emph{WIDX}$ is the highest index of a sliding register that contains a $w$-tuple.). Now, in iteration $3$, $\mathit{win}$
(read from $SLR[x_3]$ step $\sigma_3$, in line~\ref{l:rd_do_att}) 
contains a $w$-tuple (otherwise, the loop terminates after this iteration). 
Let $(\mathsf{w},j,val,h)$ be the first $w$-tuple in $win$.

Process $p_j$ writes this tuple to $SLR[x_3]$ while performing a \wrt{} operation $wop$. $wop$ therefore starts by reading $x_3$ from $M.widx$, which happens after $\sigma_2$ (by Proposition~\ref{prop:widx}, $M.widx$ is strictly increasing), and after that, reads $H[i]$ (in line~\ref{l:w_att}). This happens after $\sigma_2$ which follows $\sigma_1$.  
Hence $H[i] = x_\ell$, for some $\ell \in \{1,2,3\}$ when this read is applied 
($H[i]$ is changed to $x_k$ in iteration $k$ of the loop in line~\ref{l:rd_annouce_att}).

By Lemma~\ref{lemma:no_help}, since $x_3$ is written to $H[i]$ (line~\ref{l:rd_annouce_att}, immediately before $\sigma_3$),  
for every $x$, $x_0 < x \leq x_3$, 
$i  \notin $\textsc{readers}$(win)$, 
for any $win$ read from $SLR[x]$.  
By Lemma~\ref{lem:readers_iff_ridx}, 
this implies that $ridx_{x_3}[i] \leq x_0 < x_\ell$. 
Since any $win$ in $SLR[x_\ell]$  satisfies $i \notin $\textsc{readers}$(win)$ (Lemma~\ref{lemma:no_help}), $p_j$ inserts $i$ into its helping set $h$ (lines~\ref{l:w_att}-\ref{l:w_att_fail}). Therefore $i$ is in the helping set of the first $w$-tuple in $SLR[x_3]$, from which we conclude that $op$   terminates after the third iteration of the loop. 
\end{proof}

\subsubsection{Linearizability}
\label{sec:alg linear}

To prove linearizability, 
let $H$ be the history of \rd{}, \wrt{} and \adt{} operations 
in the execution $\alpha$.
For simplicity, we assume that 
the values written to the register in $\alpha$ are unique.
We start by classifying the operations in $H$. 
Each classified operation $op$ is also associated with an integer  $idx(op)$, which is the index of a sliding  register.

For \rd{}, we distinguish \emph{silent}, \emph{direct} and
\emph{helped} operations. 
Let  $op$ be a \rd{} operation by some process $p_i$. 
We denote by $x_0$ the value the local variable  $lsr$ when $op$ starts. 
\begin{itemize}
\item $rop$ is \emph{silent} if it is not the first \rd{} operation by $p_i$ and it immediately returns
  after reading $SLR[x_0]$  (line~\ref{l:no_new_write}). This corresponds to the case in which   no new \wrt{} operation has occurred since the last \rd{} by  $p_i$. We set $idx(op) = x_0$. 
\end{itemize}
If there is no $x > x_0$ for which $i \in \textsc{readers}(\mathit{SLR}[x])$, $op$ is \emph{unclassified}. In that case, note that   $op$ has no response in $H$. Otherwise, let $x_1 > x_0$ be the smallest index such that  $i \in $\textsc{readers}$(\mathit{SLR}[x_1])$. We set $idx(op) = x_1$ and say that  
\begin{itemize}
\item $op$ is  \emph{direct} if $i$ precedes any $w$-tuple in $SLR[x_1]$, and \emph{helped} otherwise, as in that case, $i$ appears in the helping set of the first $w$-tuple in $SLR[x_1]$.  
\end{itemize}

A \wrt($v$) operation $op$ by some process $p_j$ applies at most one write to a sliding register (line~\ref{l:w_write_SLR}).  $op$ is \emph{unclassified} if does not write to a sliding register. Otherwise, let $x$ be the index of the sliding register $op$ writes to. $x = idx(op)$ and we say that 
\begin{itemize}
\item $op$ is   \emph{visible} if $(\mathsf{w},j,v,\_)$ is the first $w$-tuple written  to $SLR[x_1]$, and \emph{hidden} otherwise. 
\end{itemize}

For \adt{}, only operations that have a response in $H$ are classified. Let $op$ be an \adt{} operation that terminates. We define $idx(op) = x$, where $x$ is the value read from $M.widx$ in the first step of $op$ (line~\ref{l:adt_read_M}). 
$op$ is \emph{non-definitive} if there is no $w$-tuple in the sequence  it reads from $SLR[x]$ (in line~\ref{l:adt_val}), 
and \emph{definitive} otherwise. 
Indeed, once a $w$-tuple has been written to $SLR[x]$, 
the set of readers of $v_{x-1}$ no longer changes, 
while  \rd{} operations may still return $v_{x-1}$ after $op$ terminates otherwise ($v_{x-1}$ the value in the first $w$-tuple in $SLR[x-1]$.).

\begin{restatable}{lemma}{snopidx}
\label{lem:sn_op_idx}
If an operation $op$ terminates before an operation $op'$ starts in $H'$,
then $idx(op) \leq idx(op')$. 
\end{restatable}

\begin{proof}
  Let  $x = idx(op)$ and $x' = idx(op')$, and let $p_i$ and $p_{i'}$ be the processes that perform $op$ and $op'$ respectively. 

  We  first prove that  when $op$  terminates, $M.widx \geq x$. As $M.widx$ is increasing (Proposition~\ref{prop:widx}), it is enough to show that $M.widx \geq x$ in a configuration in the execution interval of $op$. If $op$ is a \wrt{}  or an \adt{}, $op$ writes to $SLR[x]$ (line~\ref{l:w_write_SLR}) or reads from $SLR[x]$ (line~\ref{l:adt_val}) after reading $x$ from $M.widx$ (line~\ref{l:w_read_M} or line~\ref{l:adt_read_M}). 

  If $op$ is a direct \rd{}, $p_i$ writes $i$  to $SLR[x]$ (line~\ref{l:rd_do_att}) after having read $x$ from $M.widx$ (line~\ref{l:rd_M}). If $op$ is helped,  $M.ridx[i] =x$ when $M$ is read at the beginning of some iteration of the loop (line~\ref{l:rd_M}). Therefore $M.widx > x$, as the writeMax that changes $M.ridx[i]$ to $x$ also set $M.widx$ to $x+1$ (lines~\ref{l:rd_update_ridx1}-\ref{l:rd_help_M1}, lines~\ref{l:rd_update_ridx2}-\ref{l:rd_help_M2}, or  lines~\ref{l:w_update_ridx}-\ref{l:w_write_M}).  If $op$ is silent, it is preceded by direct  \rd{} $op''$ with the same index $x$, and hence $M.widx \geq x$ already when $op''$ terminates.

  We next examine  several cases according to the type of $op'$.
  \begin{itemize}
  \item $op'$ is a silent \rd{}. In $op'$,  $SLR[x']$ is read by $p_{i'}$ and does not contain a $w$-tuple. Hence, before $op$ starts, $\mathit{WIDX} = x'-1$ (the largest index of a sliding register to which a $w$-tuple has been written), and therefore, by Lemma~\ref{lem:phase}, $M.widx \leq x'-1$, from which we have $x \leq x'-1$, as $x \leq M.widx$ when $op$ terminates. 
  \item $op'$ is a direct or a helped  \rd{}. Let $x_0$ be the value of the local variable  $lsr$ when $op'$ starts.  $op'$ does not terminate before reading some value $x_1$ from $M.widx$ (line~\ref{l:rd_M}) in the first iteration of the repeat loop. By Lemma~\ref{lem:no_help_before_calling}, for every $y$, $x_0 < y < x_1$, the helping set of the first $w$-tuple in $SLR[y]$ does not contains $i'$. Therefore, $x_1 \leq x' = idx(op')$ and hence $x \leq x'$ as $x \leq M.widx$ before $op'$ starts. 
  \item $op'$ is a \wrt{} or an \adt{}. $p_{i'}$ reads $x'$ from $M.widx$ (line~\ref{l:w_read_M} and line~\ref{l:adt_read_M}, respectively). As $x \leq M.widx$ when $op$ terminates, $x \leq x'$. \qedhere
  \end{itemize}
\end{proof}  

We define $H'$ that contains every completed operation of $H$
as well as some incomplete operations, to which we add a matching
response. We first discard from $H$ every \adt{} and every silent \rd{}
 invocation without a matching response,  as well as every
 invocation of an unclassified \rd{} and \wrt{} operation. We then add at
 the end a response for each remaining \rd{} and \wrt{} operation that has 
 no response in $H$. The return value of a \rd{} operation $op$ with $idx(op) = x$ is the value $val$ in the  first $w$-tuple in $(\mathsf{w},\_,val,\_)$ in $SLR[x-1]$. Responses are  added in arbitrary order. 

A linearization $\lambda(\alpha)$ of $\alpha$ is defined in two steps. 
We first order operations in $H'$ according to their associated
index, in ascending order (rule $R0$). We then order operations with the
same index. Let $B(x)$ be the set of operations
$op \in H'$ such that $idx(op) = x$. These operations are ordered
according to the following rules.
\begin{enumerate}
\item[\emph{R1}] We place first silent \rd{}, direct \rd{} and non-definitive
  \adt{} operations. They are ordered according to the order in which they
  apply a read (for silent \rd{} and non-definitive \adt{}) or a write
  (for direct \rd{}) to $SLR[x]$ in $\alpha$.
\item[\emph{R2}] We then place helped \rd{} in arbitrary order,
  followed by the  definitive \adt{}
  operations. The definitive \adt{} are ordered according to the order
  in which they apply a read to $SLR[x]$ in $\alpha$.
\item[\emph{R3}] We next put every hidden \wrt{}. They are ordered according to
  the order in which their write to $SLR[x]$ is applied in $\alpha$.
\item[\emph{R4}] The (unique) visible \wrt{} is placed last. 
\end{enumerate}

Let $\lambda$ be the linearization obtained by applying
linearization rules $R0$-$R4$ to the operations in $H'$.

To show real-time order is preserved, 
we first prove the following facts about the precedence of operations with the same sequence number.

\begin{restatable}{lemma}{ordersamesn}
  \label{lem:order_same_sn}
  Let $op$, $op'$ be two operations in $H'$ with $idx(op) = idx(op') = x$.
  Let $x \geq 0$. 
  \begin{enumerate}
  \item If $op$ is  a  silent \rd{}, a direct \rd{}, or  a non-definitive
    \adt{}, and  $op'$ is 
    a helped \rd{}, a definitive \adt{}, or a \wrt{} operation,
    then $op'$      does not precede $op$. 
  \item If  $op$ is any operation and $op'$ is a \wrt{},
    then $op'$ does not precede $op$. 
  \end{enumerate}
\end{restatable}

\begin{proof}
  The proofs are as follows.
  \begin{enumerate}
  \item When $op'$, which is a helped \rd{}, a definitive \adt{}, or a \wrt{} operation,
    terminates $SLR[x]$ contains a $w$-tuple. A
    silent \rd{}, a non-definitive \adt{}, or a direct \rd{} $op$ applies a
    read (line~\ref{l:no_new_write}, line~\ref{l:adt_val})
    or write (line~\ref{l:w_write_SLR}) to $SLR[x]$ before any
    $w$-tuple  is written to it. Hence $op'$ cannot precede
    $op$. 
  \item Suppose for contradiction that $op'$ precedes $op$. When the \wrt{} operation  $op'$ terminates, $SLR[x]$ contains a $w$-tuple (line~\ref{l:w_write_SLR}) and  $M.widx \geq x+1$ (line~\ref{l:w_write_M}). Hence, every operation that starts after $op'$ and sees a $w$-tuple in the sequence held by  $SLR[x]$. Therefore, as $idx(op) = x$, $op$ cannot be a    silent \rd{}, a direct \rd{} or non-definitive \adt{}. \qedhere
  \end{enumerate}
\end{proof}

\begin{restatable}{lemma}{helpedread}
  \label{lem:helped_read}
  Let $op$ be a helped \rd{} operation in $H'$ with $idx(op) = x$. The first write of a $w$-tuple to $SLR[x]$ happens during the execution interval of $op$. 
\end{restatable}

\begin{proof} 
  As $op$ is a helped \rd{}  with $idx(op) = x$, $op$ returns after seeing that the first $w$-tuple $(\mathsf{w},\_,\_,h)$ in $SLR[x]$ is such that $i \in h$. The first write of a $w$-tuple to $SLR[x]$ therefore precedes the end of $op$. We show that this step happens after $op$ starts.  
  
  $op$ starts with $lsr = x_0$. Since $op$ is not silent, it does not terminate before entering the repeat loop. 
   In the first iteration, $x_1$ is read from $M.widx$ with $x_0 < x_1$ (Proposition~\ref{prop:att_x_increase}).
  By Lemma~\ref{lem:no_help_before_calling}, there is no $y$, $x_0 < y < x_1$ with $i \in\textsc{readers}(SLR[y])$. In particular, the first $w$-tuple written to $SLR[y]$ has no $i$ in its helping set. Hence $x_1 \leq x = idx(op)$. 

  Let $q$ be a process that writes a $w$-tuple $(\mathsf{w},\_,\_,h)$ to $SLR[x]$ with $i \in h$.
  If $x_1 < x$, we are done as $q$ first reads $x$ from $M.widx$ (which must happen after $x_1$ is read from $M.widx$ in $op$ as $M.widx$ is increasing by Proposition~\ref{prop:widx}).
  Let us assume that $x_1 = x$. Let $r_i$ and $a_i$ be respectively the value read from $M.ridx[i]$ and $H[i]$ by $q$ (in line~\ref{l:w_read_M} and line~\ref{l:w_att}, respectively). $r_i = a_i = -1$ if $op$ is the first \rd{} by $p_i$ (in which case $x_0 = -1$). Otherwise,  $i \in \textsc{readers}(SLR[x_0])$, and therefore $M.ridx[i] = x_0$ immediately after $M.widx$ is set to $x_0+1 \leq x_1$ (lines~\ref{l:rd_update_ridx1}-\ref{l:rd_help_M1}, lines~\ref{l:rd_update_ridx2}-\ref{l:rd_help_M2} or lines~\ref{l:w_update_ridx}-\ref{l:w_write_M}). 
  Hence, as for every $win$ read from $SLR[x']$, $x_0 < x' < x_1$, $i \notin $\textsc{readers}$(win)$, $M.ridx[i] = x_0 = r_i$ when it is read by $q$ (Lemma~\ref{lem:readers_iff_ridx}).

  $H[i]$ is modified by $p_i$ in line~\ref{l:rd_annouce_att} and  contains a value read from $M.widx$ at the beginning of the iteration of the repeat loop. Therefore, $a_i = x_1$ or $a_i \leq x_0$. If $a_i= x_1$, $H$ is read by $q$ after it has been changed to $x_1$, which happens in $op$, and therefore $q$ writes to $SLR[x]$ after $op$ starts. If $a_i \leq x_0$, then $q$ does not place $i$ in the helping set $to\_help$ as $a_i \leq x_0 \leq r_i$ (line~\ref{l:w_help_wd}), contradicting the assumption that $q$ writes to $SLR[x]$ a tuple $(\mathsf{w},\_,\_,h)$ with $i \in h$. \end{proof}

Next lemma shows that real time order among operations is preserved: 

\begin{restatable}{lemma}{rtorder}
\label{lem:rt_order}
If an operation $op$ terminates before an operation $op'$ starts in $H'$,
then $op$ precedes $op'$ in $\lambda$. 
\end{restatable}
\begin{proof}
  By Lemma~\ref{lem:sn_op_idx}, $idx(op) \leq idx(op')$. If $idx(op) <
  idx(op')$, $op$ is before $op'$ in $\lambda$ by rule $R0$. We assume
  in the following that  $idx(op) = idx(op') = x$.
  
  \begin{itemize}
  \item If $op'$ is a silent \rd{}, a direct \rd{}, or non-definitive
    \adt{}, it follows from Lemma~\ref{lem:order_same_sn}(1)
    that $op$ also falls into this category. Therefore, $op$ and
    $op'$ are both ordered in $\lambda$ using rule $R1$. They are ordered
    according to the order in which a step in their execution
    interval occurs in $\alpha$. Hence, $op$ precedes $op'$ in $\alpha$ implies
    that $op$ precedes $op'$ in $\lambda$.
  \item If $op'$ is a helped \rd{} or a definitive \adt{}, it follows
    from Lemma~\ref{lem:order_same_sn}(1) and Lemma~\ref{lem:order_same_sn}(2) that $op$
    also falls into this category, or is a silent or direct \rd{}, or a
    non-definitive \adt{}. In the latter case, $op$ is ordered in
    $\lambda$ according to rule $R1$, and $op'$, rule $R2$, from which
    we have that $op$ precedes $op'$ in $\lambda$.

    In the former case, as $op$ and $op'$ are both helped $\rd{}$, their execution
    interval intersect (Lemma~\ref{lem:helped_read}) and thus $op$
    cannot precedes $op'$. If $op$ and $op'$ are both definitive
    \adt{}, they are placed according to the order in which they apply
    a read to $SLR[x]$ in $\alpha$. Hence $op$ is before $op'$ in $\lambda$.
    If $op$ is an helped \rd{} and $op'$ a definitive \adt{}, $op$ is
    placed before $op'$ in $\lambda$ by rule $R2$. The last case
    remaining is $op$ being a definitive \adt{} and $op'$, a helped
    \rd{}. As $op$ is a definitive \adt{}, it sees a $w$-tuple 
    mark in $SLR[x]$, but the step in which the first such tuple
    is written to $SLR[x]$ is in the execution interval of $op'$
    (Lemma~\ref{lem:helped_read}). Therefore $op$ cannot terminate
    before $op'$ starts. 
  \item If $op'$ is a \wrt{} visible or hidden,
    Lemma~\ref{lem:order_same_sn}(2) implies that $op$ cannot be a
    \wrt. Therefore, $op$ is placed according to rule $R1$ or $R2$,
    and sthus precedes $op'$ which is placed after, according to $R3$
    or $R4$. \qedhere
  \end{itemize}
\end{proof}

We next show that any \rd{} returns the last value written that precedes it in $\lambda$. 
\begin{restatable}{lemma}{readlast}
\label{lem:read_last}
If \rd{} operation $op$ in $H'$ returns $v$,
then $v$ is the value written by the last \wrt{} that precedes $op$ in $\lambda$, or the initial value $v_{0}$ if there is no such \wrt{}.
\end{restatable}
\begin{proof}
Let $x = idx(op)$, and let $p_j$ be the process that performs $op$.

We first consider the case $x > 0$. If $op$ is helped or direct, the value returned by $op$ is
  the value $v$ in the first $w$-tuple $(\mathsf{w},i,v,\_)$ stored in $SLR[x-1]$
  (lines~\ref{l:getValue}). Note that $i$  and $v$ are 
  well defined, as when $p_j$ reads from or writes to $SLR[x]$,
  a $w$-tuple has already be written to $SLR[s-1]$
  (Proposition~\ref{prop:lower_invalidated}). Let $wop$ be the operation $p_i$ is
  performing when it writes  $(\mathsf{w},i,v,\_)$ to $SLR[x-1]$. By definition, $idx(op) =
  x-1$, and $wop$ is a visible \wrt{}. 
  Therefore, among the operations with index  $x-1$, $wop$ is placed last in $\lambda$ (rule $R4$), and there is no other \wrt{} operation between  $wop$ and $op$ (every \wrt{} operation with index $\geq x$ is after $op$ in $\lambda$.)
  
  If $rop$ is silent, it is preceded by a direct \rd{} operation $op'$
  by the same process with the same index $x$. By the
  linearization rule $R1$, there is no \wrt{} operation between $op'$
  and $op$ in $\lambda$. Also, $op$ and $op'$ return the same
  value. By the same reasoning as above, it follows that $op$ returns
  the input value of the last \wrt{} operation that precedes it in
  $\lambda$.

  If $x=0$, $op$ returns the initial value $v_0$. Indeed, $SLR[-1]$ is
  initialized with a sequence that contains  a single tuple 
  $(\mathsf{w},i_0,v_0,\emptyset)$. There is no \wrt{} that precedes $op$ in $\lambda$, as for every  \wrt{} operation $wop$,  $idx(wop) \geq 0$. 
\end{proof}

Finally, we show 
(Lemma~\ref{lem:adt_complete} and Lemma~\ref{lem:adt_accurate}) 
that a pair $(p,v)$  is in the set returned by an  \adt{} operation $op$ if and only if there is a \rd{} by $p$ returning $v$ precedes $op$. 

\begin{lemma}
  \label{lem:read_record_once}
  If $i \in \textsc{readers}(SLR[x])$ 
  then there exists a \rd{} operation $op$ in $H'$ by process $p_i$ with $idx(op) = x$.
\end{lemma}

\begin{proof}
$\implies$ Suppose that $i$ precedes any $w$-tuple in $SLR[x]$. $p_i$ writes $i$ to $SLR[x]$ (line~\ref{l:rd_do_att}) while performing the $k$th iteration of the repeat loop (lines~\ref{l:rd_start_repeat}-\ref{l:rd_att_success}) in some \rd{} operation $op$. 
  Let $x_0$ be the value of $lsr$ when $op$ starts, and let $x_1 < \ldots < x_k$ be the value read from $M.widx$ (line~\ref{l:rd_M}) in the first $k$ iterations of the loop. Note that $x = x_k$. 
  As before writing $i$ to $SLR[x_k =x]$, $H[i]$ is changed to $x_k$ (line~\ref{l:rd_annouce_att}), it follows from Lemma~\ref{lemma:no_help} that  for every $x', x_0 < x' < x_k$, $i \notin \textsc{readers}(SLR[x'])$. As $i \in \textsc{readers}(SLR[x_k])$, $x_k = \min\{x' > x_0: i \in $\textsc{readers}$(SLR[x'])\}$ and therefore by definition $x = x_k = sn(op)$.

    Suppose now that $i$ is in the helping set of the first $w$-tuple $(\mathsf{w},j,\_,h)$ written to $SLR[x]$. Before writing $(\mathsf{w},j,\_,h)$ (line~\ref{l:w_write_SLR}), $p_j$ reads $x$ from $M.widx$ , $y$ from $M.ridx[i]$ (line~\ref{l:w_read_M}) and $x_\ell$ from $H[i]$ (line~\ref{l:w_att}). As $i$ is placed into $to\_help$, $y < x_\ell$.

    $H[i]$ is changed by $p_i$ (line~\ref{l:rd_annouce_att}) in an iteration of the repeat loop  while performing some \rd{} operation $op$. Let $x_0$ be the value of $lsr$ when $op$ starts,  and  $p_i$ is performing iteration $\ell$ when it writes $x_\ell$ to $H[i]$. 
    By the code, $x_\ell$ is the value read from $M.widx$ in that iteration, and by Lemma~\ref{lemma:no_help}, $i \notin \textsc{readers}(SLR[x'])$, for every $x', x_0 < x' < x_\ell$.

    To summarize, starting from $i \in $\textsc{readers}$(SLR[x])$, we have shown that there exists a \rd{} operation $op$ by $p_i$ that starts with  $lsr = x_0$, and that for every $x', x_0 < x' < x (=x_\ell)$, $i \notin$\textsc{readers}$(SLR[x'])$. By definition, $idx(op)  =x$. \qedhere
\end{proof}

\begin{restatable}{lemma}{adtcomplete}
  \label{lem:adt_complete}
  Let $aop$ be an \adt{} operation in $H'$ that returns $A$. If there
  is a \rd{} operation $rop$ by process $p_j$ returning $v$ and that
  precedes $aop$ in $\lambda$, $(j,v) \in A$. 
\end{restatable}
\begin{proof}
  By linearization rules $R0$-$R2$, $idx(rop) = x \leq idx(aop)$. Let
  $p_i$ be the process that performs the \adt{} operation $aop$.

  If $rop$ is silent, it is preceded by a direct \rd{} operation by
  the same process, with the same output value $v$ and the same
  index $x$. In the following, we thus assume that $rop$ is
  direct or helped. In $SLR[x]$, $j$  is written before any $w$-tuple, or the first $w$-tuple  written to $SLR[x]$  has an helping set $h$ containing $i$. Note that $v$ is the value of the first $w$-tuple in $SLR[x-1]$.

  Let us assume that $x < idx(aop) = x'$. In $aop$, $p_i$ reads $x'$ from $M.widx$  and a set  $as'$ from $M.auditset$ (line~\ref{l:adt_read_M}). Before this step, $M$ is changed from $x$ to $x+1$ (in line~\ref{l:rd_help_M1}, line~\ref{l:rd_help_M2} or line~\ref{l:w_write_M}), after a $win$ containing a $w$-tuple is read from  $SLR[x]$. Hence, $j \in  \textsc{readers}(win)$, and therefore  $(j,v)$ is added to the audit set $as$ written together with $x+1$ to $M$. 
  By the code (line~\ref{l:rd_update_ridx1}, line~\ref{l:rd_update_ridx2} or line~\ref{l:w_update_ridx} and line~\ref{l:adt_current_val}), we have that
  $as \subseteq as' \subseteq A$.

  We now assume that $x=idx(aop)$. If $aop$ is a non-definitive \adt{}, no $w$-tuple has been written to $SLR[x]$ when $p_i$ reads $SLR[x]$. As $rop$ precedes $aop$ in $\lambda$, there are both placed in $\lambda$ following rule $R1$ and therefore $p_j$ writes $j$ to $SLR[x]$ before the sliding register is read in $aop$, from which it follows that $(j,v) \in A$ (line~\ref{l:adt_current_val}). Otherwise, $aop$ is a definitive \adt{}, which means that the value $win$ reads from $SLR[x]$ (line~\ref{l:adt_current_val}) contains the  first $w$-tuple $(\mathsf{w},\_,\_,h)$   written to $SLR[x]$. As $j$ precedes this tuple or $i \in h$, $j \in \textsc{readers}(win)$ and therefore $(j,v) \in A$. 
\end{proof}

\begin{restatable}{lemma}{adtaccurate}
  \label{lem:adt_accurate}
  Let
  $aop$ be an \adt{} operation in $H'$ whose response contains
  $(j,v)$. There exists a \rd{} operation by $p_j$ returning $v$ that
  precedes $aop$ in $\lambda$. 
\end{restatable}
\begin{proof}
  Let $p_i$ be the process that performs $aop$, and let $A$ be the set of pairs (process,value) returned by this operation. As $(j,v)$ is in $A$, there exists $x$ such that $i \in $\textsc{readers}$(SLR[x])$ and $v$ is the value in the first $w$-tuple $(\mathsf{w},\_,v\_)$ written to $SLR[x-1]$. Indeed, pair $(j,v)$ is inserted to $A$ by $p_i$ after reading $SLR[x]$, if $x = idx(aop)$ (line~\ref{l:adt_current_val}), or when $M.widx$ is changed to $x+1$ (lines~\ref{l:rd_help_A1}-\ref{l:rd_help_M1}, lines~\ref{l:rd_help_A2}-\ref{l:rd_help_M2} or lines~\ref{l:w_A}-\ref{l:w_write_M}) if $x < idx(aop)$. 
  
  Lemma~\ref{lem:read_record_once} shows there is \rd{} operation $op$ by process $p_i$ with $idx(op) = x$. This operation returns $v$, which the value of the first $w$-tuple written to $SLR[x-1]$. 

  It remains to prove that $op$ is before
  $aop$ in $\lambda$.
  If $x < idx(aop)$, $op$ precedes $aop$ in $\lambda$ (rule $R0$). Otherwise, $x =idx(aop)$.
  If $aop$ is definitive, it is linearized after $rop$ by rules $R1$ and
  $R2$. Otherwise, $aop$ is non-definitive, and  
  $j$ thus appears before any  $w$-tuple in $SLR[x]$. $j$ is also written to $SLR[x]$ before 
  before $SLR[x]$ is read by $p_i$.  $op$ is therefore  direct, and
  linearized before $aop$ by rule $R1$. 
\end{proof}


\section{Auditable LL/SC from $2n$-Sliding Register}
\label{sec:LLSC}

We show how to adapt our auditable register algorithm to implement an auditable LL/SC object for $n$ processes, and an arbitrary number of auditors. An \adt{} operation returns a set of pairs $(p,v)$, representing the values returned by the processes by the \llk{} operations preceding the \adt{}. 
The implementation uses $2n$-sliding registers.  
The main result of this section is: 

\begin{theorem}
  \label{thm:llsc_auditable}
  There is a wait-free, linearizable implementation of an $n$-process auditable LL/SC object from $2n$-sliding registers and standard registers with $O(n)$ step complexity  per operation. 
\end{theorem}


Algorithm~\ref{alg:llsc_auditable} is essentially the same as Algorithm~\ref{alg:mw_auditable}. \llk{} operations are identified with \rd{} operations, and \scd{} with \wrt{}. At the end of a finite execution $\alpha$ of Algorithm~\ref{alg:mw_auditable}, the sequences stored in the sliding registers in the array  $SLR$ indicate, for each $x$, which is the  $x$th value $v_x$ held in the auditable register, and which processes write or read this value. Specifically, given the sequence $win_x$ stored in $SLR[x]$, the first $w$-tuple $(\mathsf{w},j,v,h)$ in $win_x$ indicates that $v_x = v$, and its writer is $p_j$. The readers of $v_{x-1}$ are the processes $p_i$, where  $i \in h$, or precedes $(\mathsf{w},j,v,h)$ in $win_x$. Each other $w$-tuple $(\mathsf{w},j',v',\_)$ in $win_x$ corresponds to an \emph{invisible} write operation, whose input $v'$ is never read.
We may think as these operations as \emph{unsuccessful}, in the sense that they fail to change the value of the auditable register, being immediately overwritten by another write.

\begin{algorithm}[tb]
\small
  \caption{$n$-process LL/SC with auditable LL}
  \label{alg:llsc_auditable}
  \begin{algorithmic}[1]
    \State\textbf{shared variables}
    \State\hspace{\algorithmicindent}$\mathit{M}, H[1..n]$:  max register and helping array of $n$ SWMR register, initialized as in Algorithm~\ref{alg:mw_auditable}
    \State\hspace{\algorithmicindent} $SLR[-1,0,\ldots]$: unbounded array of $2n$-sliding registers, initialized as in Algorithm~\ref{alg:mw_auditable}
    \State\textbf{local variables}    
    \State\hspace{\algorithmicindent} $lsr, val$, as in Algorithm~\ref{alg:mw_auditable}

    \Function{ll}{$~$} \Comment{identical to \rd{} in Algorithm~\ref{alg:mw_auditable}}
    \EndFunction
    
    \Function{sc}{$v$} \Comment{code for process $p_i, i \in \{1,\ldots,n\}$}
    \State $widx,ridx,auditset  \gets M.\mathsf{read}()$;  $to\_help  \gets \emptyset$ \label{l:sc_read_M}
    \State \colorbox{gray!20}{\textbf{if} $widx > lsr$ \textbf{then} \Return{\emph{false}}} \Comment{a successful \scd{} happened since $p_i$'s last \llk{}} \label{l:sc_return1}
    \ForAll{$i \in \{1,\ldots,m\}$} $aidx_j \gets H[j].\mathsf{read}()$ \Comment{help ongoing \llk{} operations}\label{l:sc_att}
    \If{$ridx[j] <  aidx_j$} \Comment{ an \llk{} by $p_j$  may need help} \label{l:sc_help_needed?}
    \State $window \gets SLR[aidx_j].\mathsf{read}()$  \label{l:sc_help_wd}
    \State \textbf{if} $j \notin \Call{readers}{window}$ \textbf{then} $to\_help \gets to\_help \cup  \{j\}$ \label{l:sc_att_fail}
    \EndIf
    \EndFor
    \State $SLR[widx].\mathsf{write}(\mathsf{w},j,v,to\_help)$;\Comment{announce new \scd{}}\label{l:sc_write_SLR}
    \State $window \gets SLR[widx].\mathsf{read}()$; $val \gets \Call{getvalue}{widx-1}$
    \State \textbf{for each} $j \in \Call{readers}{window}$  \textbf{do} 
    $ridx[j] \gets widx$; $auditset \gets auditset \cup \{(j,val)\}$ \label{l:sc_update_ridx}\label{l:sc_A}
    \State $M.\mathsf{writeMax}(widx+1,ridx,auditset)$; \label{l:sc_write_M}
    \State \colorbox{gray!20}{\textbf{if} $(\mathsf{w},j,v,to\_help)$ is the 1st $w$-tuple in $window$ \textbf{then} \Return{\emph{true}} \textbf{else} \Return{\emph{false}}} \label{l:sc_return2}
    \EndFunction
    \Function{audit,getValue, readers}{} \Comment{as in Algorithm~\ref{alg:mw_auditable}}
    \EndFunction
  \end{algorithmic}
\end{algorithm}

Alternatively, we may think of  the sequences in $SLR$ as a trace of an execution $\beta$ of an LL/SC object implementation by identifying reads with \llk{} and writes with \scd{}. Again, the $x$th value held by the object in $\beta$ is $v_x$, the value in the first $w$-tuple $(\mathsf{w},j,v,h)$ in $win_x$. Each process $p_i$, with $ p_i \in h$ or preceding this tuple has an \llk{} that returns $v_{x-1}$. The first $w$-tuple $(\mathsf{w},j,v,h)$ indicates a successful $\scd{}(v)$ by process $p_j$, that changes the object from $v_{x-1}$ to $v_x = v$. And every following  tuple $(\mathsf{w},j',v',\_)$ marks  an unsuccessful $\scd{}(v')$ by process $p_{j'}$. Recall that, per its specification,  an \scd{} is successful if and only if it is preceded by an \llk{} by the same process, without any successful \scd{} operation  between them. In particular, $\beta$ is a valid sequential execution  if (1) in each sequence $win_x$, the first $w$-tuple $(\mathsf{w},j, v,h)$ is after $p_j$'s \llk{} and (2) each \scd{}($v'$) operation corresponding to a following $(\mathsf{w},j',v',\_)$ tuple is after \scd{}($v$) in $\beta$.

The code of an  $\llk{}$ operation is the same as for a \rd{} operation in Algorithm~\ref{alg:mw_auditable}, 
while \adt{} and the auxiliary functions \textsc{getValues} 
and \textsc{readers} are identical.  
\scd{} operations follow the code of \wrt{}, 
with two additions (line~\ref{l:sc_return1} and line~\ref{l:sc_return2})  
\colorbox{gray!20}{highlighted in gray}. 

To maintain property (1), an \scd{} operation $op$ by process $p_i$ should be prevented from writing a $w$-tuple to a sliding register $SLR[x]$ in which $p_i$ 's last \llk{} is not recorded. 
In Algorithm~\ref{alg:mw_auditable},  the local variable $lsr$ of a process $p_i$  is the  highest index of a sliding register that keeps track of $p_i$'s last \rd{}, and hence here, $SLR[lsr]$ records $p_i$'s last \llk{}. Therefore, before writing to $SLR[widx]$ (in line~\ref{l:sc_write_SLR}), $p_i$ checks that $lsr = widx$. If this is not the case, a successful \scd{} has occurred since $p_i$'s last \llk{}, and $op$ may immediately return \emph{false} (line~\ref{l:sc_return1}). The other addition is the return statement at line~\ref{l:sc_return2}: \emph{true} is returned if $p_i$'s $w$-tuple  is the first in $SLR[widx]$, and $\emph{false}$ otherwise.

For property (2), the linearization rules are slightly modified. In Algorithm~\ref{alg:mw_auditable}, \wrt{} operations recorded  in the same  sliding register $SLR[x]$ are linearized in the reverse order of their corresponding $w$-tuple appearance in the sequence stored in $SLR[x]$. Here, we do the opposite, linearizing \scd{} operations in the same order their corresponding $w$-tuple appear in $SLR[x]$. Only the first $w$-tuple represents a successful \scd{}, which is aligned with the return statement of line~\ref{l:sc_return2}.

Finally, as each process may write twice to a sliding register, once in an \llk{} operation, and once in a \scd{} operation, $SLR$ is an array of $2n$-sliding registers. The code of \adt{} may easily be adapted to report which process has successfully stored which value instead of or in complement to values returned by \llk{} operations to  the processes. 

The code of \adt{} is the same in Algorithm~\ref{alg:mw_auditable} and Algorithm~\ref{alg:llsc_auditable}, and the same code is shared by \rd{} and \llk{} operations. 
For \scd{}, the additional statements (line~\ref{l:sc_return1} and line~\ref{l:sc_return2}) in the code do not affect termination or step-complexity. Therefore,
the step complexity of \llk{}, \scd{}, and \adt{} is $O(n)$. 

Fix a finite \emph{well-formed} execution $\beta$, in which 
each process alternates \llk{} and \scd{} operations. Thanks to the similarities in the code, a linearization $\mu(\beta)$ of $\beta$ can be obtained from a linearization $\lambda(\alpha)$ of an execution $\alpha$ of the auditable register implementation induced by $\beta$. $\alpha$ is constructed as follows. 
We introduce a new class for \scd{} operations that terminate immediately after reading $M.widx$ in line \ref{l:sc_return1}. Those operations are said to be \emph{silent} (similarly to silent \rd{} operations). We next extract from  $\beta$ an execution $\alpha$ of Algorithm~\ref{alg:mw_auditable} by removing all steps applied by silent \scd{}, and replacing each invocation of \scd{} and \llk{} by and invocation of \wrt{} and \rd{}, respectively, with the same input.  $\alpha$ is a valid execution of Algorithm~\ref{alg:mw_auditable}, as besides early termination for \scd{} (line~\ref{l:sc_return1}), which we dispose of by removing steps of silent \scd{} operations, the code of \scd{} and \llk{} is the same as the code of \wrt{} and \rd{}, respectively, and the code for \adt{} is identical.

We  next show how to construct a linearization $\mu(\beta)$ of $\beta$ 
from the linearization  $\lambda(\alpha)$ of $\alpha$. 
\begin{enumerate}
\item
  \wrt{} operations  with the same index $x$ are reordered, by applying first rule $R4$ and then rule $R3$. Hence, the \wrt{} operation corresponding to the first $w$-tuple in $SLR[x]$ is placed first, and then the $\wrt{}$s corresponding to the other $w$-tuples follow, in arbitrary order. The resulting sequence $\lambda'(\alpha)$ is no longer a valid linearization of an auditable  register, but still extends the real-time order among operations, as \wrt{}s  with the same index are concurrent (Lemma~\ref{lem:order_same_sn}(2)). Observe  that in $\lambda'(\alpha)$, as  in $\lambda(\alpha)$, \wrt{}s with the same index form a contiguous block, denoted $W_x$. A \rd{} operation no longer returns the input of the last preceding \wrt{}, but rather the input $v_x$ of the first \wrt{} in the last block $W_x$ that precedes it.
\item
  We revert to \llk{}/\scd{} operations, by replacing in $\lambda'(\alpha)$ each \rd{} and \wrt{} operation by the \llk{} or \scd{} operation it originated from. This results in a partial linearization $\mu'(\beta)$ of $\beta$, missing silent \scd{} operations.
  In $\mu'(\beta)$, an \llk{} operation $op$ returns  the input $v_x$  of the first \scd{} operation $sop$ in the last block $W_x$ that precedes it. 
  Every \scd{} operation in $W_x$ writes a $w$-tuple to $SLR[x]$, $sop$ being the first to do so. Hence, except $sop$, each of them returns false (line~\ref{l:sc_return2}). 
  Consequently,  $op$ returns the input of the last preceding successful \scd{}.

  If $p_i$ has a \scd{} operation in $W_x$, its matching preceding \llk{} operation $op$ is recorded in $SLR[x]$ (line~\ref{l:sc_return1}) and therefore, has index $idx(op) =x$. 
  By  rules $R1$-$R2$, $op$ is before $W_x$, and after any operation $op'$ with index $idx(op') < x$ (rule $R0$). Only the  first \scd{} operation $sop$  in $W_x$ returns true, and indeed, there is no \scd{} between $sop$ and its last preceding \llk{}. 
  Each other \scd{} operation in $W_x$ is unsuccessful, in agreement with the fact that it preceded by the successful \scd{} $sop$, with no \llk{} operation in between. 

  \llk{} and \adt{}  are ordered in  $\mu'(\beta)$  as their corresponding \rd{} and \adt{} in $\lambda(\alpha)$, and  each \llk{}  returns the same value as its corresponding \rd{}. Hence, a pair $(p,v)$ is included in the response of an \adt{} operation $op$ if and only if there is an \llk{} by $p$ that returns $v$ before $op$ in $\mu'(\beta)$. 

\item Finally, we bring back the  silent \scd{} operations that were removed from $\beta$ to form a full linearization $\mu(\beta)$. Each silent \scd{} operation $op$ that returns  is unsuccessful and therefore can be inserted in $\mu'(\beta)$ without affecting the outcome of other \scd{} or \llk{} operations, as long as there is a successful \scd{} between 
the matching \llk{} $op'$ by the same process and $op$.

  Let $x$ be the value read by $op$ from $M.widx$. 
  We insert $op$ in $\mu'(\beta)$ within the operations whose index is $x$ and after every operation that terminates before $op$ starts. As $x > x'$, $op$ is placed after the operations with index $x'$. In particular, $op$ follows  the successful  \scd{} operation $sop$ that writes the first $w$-tuple to $SLR[x]$, which in turn follows $op'$. This can be done while preserving the real-time order, as when $x>x'$ is read from $M.widx$ in $op$, a $w$-tuple has already been written to $SRL[x']$ (Lemma~\ref{lem:phase}). 
  \end{enumerate}


\section{Immediate Deny List from Auditable Registers}
\label{label:CN_of_DL}

A \emph{Deny List}~\cite{FreyGR23} is used to control resources, 
by having a set of managers maintain a list of which users 
are unauthorized to access which resources.
To access a resource, a user must prove that the corresponding user-resource 
pair has not been added to the Deny List. 
The managers have to agree on the set of process-resource pairs in the Deny List. 
A Deny List has an \emph{anti-flickering} property that ensures that there are
no transient periods: once access is disallowed, it is never allowed again. 

More formally, a 
\emph{deny list} object over 
a set of resources $S$ supports three operations, for $x \in S$: 
$\textsc{append}(x)$, 
$\textsc{prove}(x)$ which returns a Boolean value, 
and $\textsc{read}()$, which returns a set of process-resource ($p,x$) pairs.  A $\textsc{prove}(x)$ that returns \emph{false} is \emph{invalid};
otherwise, it is \emph{valid}.
The intuition is that an $\textsc{append}(x)$ revokes the authorization to access a resource $x$ to all processes. A valid $\textsc{prove}(x)$ by process $p$ indicates that $p$ is authorized to access $x$.
The set of processes that can invoke $\textsc{append}$ is called the \emph{managers}, and those that can invoke $\textsc{prove}$ are called the \emph{provers}. These sets of processes are predefined and static.
The property \emph{termination} requires that the operations $\textsc{prove}$,
$\textsc{append}$, and $\textsc{read}$ return within a finite number of steps.

In the original definition~\cite{FreyGR23}, 
an $\textsc{append}$ can be \emph{successful} or \emph{unsuccessful}. 
Their sequential specification in the \emph{anti-flickering flavor} 
of a deny list is as follows:
\begin{description}
   \item[Append progress:] Only a finite number of $\textsc{append}(x)$  operations are unsuccessful;  that is, $\textsc{append}(x)$ is eventually successful. 
   \item[Prove progress:] After a successful $\textsc{append}(x)$ operation, only a finite number of $\textsc{prove}(x)$ operations can be valid; that is,  $\textsc{prove}(x)$ is eventually invalid after a successful $\textsc{append}(x)$.
   \item[Prove validity:] A $\textsc{prove}(x)$ operation is invalid only if a successful $\textsc{append}(x)$ operation appears before it. 
    \item[Prove anti-flickering:]
    If a $\textsc{prove}(x)$ operation $op$ is invalid, then all following $\textsc{prove}(x)$ operations are invalid. 
    \item[Read validity:] 
    The set of object-process pairs returned by a $\textsc{read}()$ operation 
    includes exactly all preceding valid $\textsc{prove}(x)$ operations 
    and the processes that invoked them. 
\end{description}

We consider a stronger version, called \emph{immediate} deny list,  where all $\textsc{append}(x)$ are successful, and where all $\textsc{prove}(x)$ operations that follow an $\textsc{append}(x)$ are invalid. The sequential specification for the \emph{immediate} deny list is therefore as follows: 
\begin{description}
\item[Strong prove validity:] A $\textsc{prove}(x)$ operation is invalid if and only if an $\textsc{append}(x)$ operation appears before it. 
\item[Read validity:] 
    The set of object-process pairs returned by a $\textsc{read}()$ operation includes exactly all preceding valid $\textsc{prove}(x)$ operations 
    and the processes that invoked them. 
\end{description}

An \emph{Immediate Deny List} guarantees all the properties of the \emph{Anti-flickering Deny List}:
Since all $\textsc{append}(x)$ are successful, we trivially guarantee  \emph{Append progress}. Our \emph{Strong prove validity} includes the \emph{Prove validity} and \emph{Prove progress} where the number of $\textsc{prove}(x)$ that can be valid after an $\textsc{append}(x)$ is $0$. 
\emph{Prove anti-flickering} is trivially implied by our \emph{Strong prove validity} property.
 
It is already known that the consensus number of an \emph{Anti-flickering Deny List} object where $n$ processes can both do $\textsc{append}(x)$ and $\textsc{prove}(x)$, denoted anti-flickering $n$-deny list, is $n$~\cite{FreyGR23}. In particular, 
they show how to solve consensus among $n$ processes using this object with a single resource. This implies that the Anti-flickering $n$-deny list has at least consensus number $n$.
Their result implies that the consensus number of the immediate $n$-deny list is at least $n$.

They also present an algorithm to show that the consensus number of the anti-flickering $n$-deny list over a set of resources $S$ is at most $n$. In the following, we show how to implement (Algorithm~\ref{alg:dl_from_ar}) an immediate-$n$-deny list object using \adReg{1}{n-1}s which have consensus number $n$ as proved in Sections~\ref{sec:mwa-cn-lower-bound} and~\ref{sec:fast_mwmr_auditable}.
This proves that the immediate $n$-deny list has consensus number at most $n$.
Our algorithm provides a somewhat simpler proof for their upper bound. 
We present the algorithm for a single resource. The generalized version can be easily built thanks to the locality property of linearizability, and by implementing the $\textsc{read}()$ on the set of resource $\m{S}$ using classical techniques to implement a snapshot by applying the $\textsc{read}()$ on each $x \in \m{S}$.

\begin{algorithm}[tb]
\small
  \caption{Immediate $n$ deny-list object from auditable registers, code for process $p_i$}
  \label{alg:dl_from_ar}
  \begin{algorithmic}[1]
    \Statex\textbf{Shared objects}
    \Statex\hspace{\algorithmicindent} 
    $AR_x[1\ldots n]$ is a vector of  $(1,n-1)$-auditable registers, initially \emph{true}
    \State\textbf{local variables}    
    \State\hspace{\algorithmicindent} 
    $append_x$ is a boolean initially $\mathit{false}$, that is set to $true$ when $p_i$ does $\textsc{append}(x)$
    
    \Function{append}{$x$}
    \State $AR_x[i].\mathsf{write}(\mathit{false})$; $append_x \leftarrow{} true$
    \Return{}
    \EndFunction

    \Function{prove}{$x$}
    \If{$append_x$} \Return{$\mathit{false}$};\label{return-for-previous-append}\EndIf
    \ForAll{$j \in \{1,\ldots,n\}$ with $j\neq i$}
    \State $b\leftarrow{} AR_x.\mathsf{read}()$; \textbf{if} $\neg b$ \textbf{then} \Return{$\mathit{false}$}
    \EndFor
    \State\Return{$true$}
    \EndFunction

    \Function{read}{$~$} \Comment{audit all registers until you get a successful double collect}
    \State $c2 \gets \emptyset$ 
    \Repeat{}
    \State $c1 \gets c2$; $c2 \gets \emptyset$ \label{algo-ndeny-setc1}
    \State\textbf{for all} $j \in \{1,\ldots,n\}$ \textbf{do} 
    $a_j \gets AR_x[j].\mathsf{audit}()$
    \State $c2 \gets \{(q,x) : \forall j\neq q, (q,\emph{true}) \in a_j\}$
    \Until{$c1 = c2$}
    \State \Return{$c2$}
    \EndFunction
  \end{algorithmic}
\end{algorithm}

Algorithm~\ref{alg:dl_from_ar} implements an immediate $n$-deny list for a resource $x$. It  uses a vector $AR_x[1\ldots n]$ of binary auditable registers, initially \emph{true}. $AR_x[i]$ is written by process $p_i$ and read by all other processes.
When a process $p_i$ wants to perform $\textsc{append}(x)$, it writes false in its auditable register $AR_x[i]$ and locally stores the information that it did an $\textsc{append}$.

 To perform $\textsc{prove}(x)$, a process  simply checks if it previously did an $\textsc{append}(x)$, and if not, it reads all the auditable registers (but its own) to check if some process wrote into one such register. If none of the previous conditions happen, then it can return $true$.

The $\textsc{read}()$ repeatedly audits all the registers to collect the processes that have executed a valid $\textsc{prove}(x)$. These are the ones that read $true$ in all registers. To obtain a snapshot, the $\textsc{read}()$ terminates when two consecutive collects returns the same set. 

To prove the correctness of Algorithm~\ref{alg:dl_from_ar}, 
we start by noting that all operations terminate within a finite 
number of their own steps.
This is obvious for $\textsc{append}$ and $\textsc{prove}$, 
relying on wait-freedom of the encapsulated auditable register operations.
For a $\textsc{read}$ operation, note that the sets it returns are 
monotonically contained in each other;
the maximal set is the one containing all process-object pairs. 
Furthermore, if two sets are equal, the operation terminates.
Thus, the loop can be repeated at most $n$ times.

We next prove that Algorithm~\ref{alg:dl_from_ar} implements a linearizable immediate $n$-deny list.
First, for all $i \in \{1,\ldots, n\}$, we linearize all the $\textsc{append}(x)$ by $p_i$ in the step corresponding to the write to the auditable register. We complete only the $\textsc{append}(x)$ operations that have done this step.

A $\textsc{prove}(x)$ operation by $p_i$ that returns $\mathit{false}$ and that does not return at line \ref{return-for-previous-append} is linearized at the last read primitive $p_i$ applies to an auditable register $AR_x[j]$. A $\textsc{prove}(x)$ that returns at line \ref{return-for-previous-append} is linearized at its invocation.
A $\textsc{read}$ is linearized at the last audit of the second last loop (i.e., the last audit before $c_1$ is set for the last time). We remove incomplete $\textsc{read}$.

We finally add to the linearization all the $\textsc{prove}(x)$ operations that return $true$. They are inserted at the beginning before the first $\textsc{append}(x)$, according to their real-time order. In particular, we pick one by one in the order of their last step (i.e., the last read from an auditable register), the ones that happen earlier first), denoted $s$. We put $\textsc{prove}(x)$ immediately before the first $\textsc{read}()$ whose linearization point follows $s$ or immediately the first $\textsc{append}(x)$ if such $\textsc{read}()$ does not exist.

\emph{Strong prove validity:} A $\textsc{prove}(x)$ by a process $p$ returns $\mathit{false}$, either if $p$ previously performed an $\textsc{append}(x)$ or if it read $\mathit{false}$ in one auditable register. 
In both cases, there is an $\textsc{append}(x)$ in the set of operations we linearize and it is linearized before the $\textsc{prove}(x)$ according to our rules. The only if part is immediate since we linearize all $\textsc{prove}(x)$ that return $true$ before the first $\textsc{append}(x)$.

\begin{lemma} 
The set of object-process pairs returned by a $\textsc{read}()$ operation includes exactly all preceding valid $\textsc{prove}(x)$ operations and the processes that invoked them. 
\end{lemma}

\begin{proof}
A $op=\textsc{prove}(x)$ that returns $true$ is linearized before the first audit, namely $op'$, whose linearization point follows the last read primitive applied by $op$ to an auditable register. After the step where it is linearized $op'$ executes a second loop where it audits all the auditable registers. At that point $op$ is added in $c_2$. By the definition of the linearization point of $op'$, $op$ is in the set returned by $op'$.

On the other hand, suppose that a pair $(q,true)$ is in the set $c_2$ returned by a $\textsc{read}()$ by $p$. Then, there is a $\textsc{prove}(x)$ by $q$ that read $true$ in all the auditable registers (but $AR_x[q]$) before $p$ start the execution of the last loop. By a simple inspection of the pseudo code it is easy to see that $op$ returns $true$. Our linearization rule completes the proof.
\end{proof}

Finally, we show that our linearization respects the \emph{real-time order} of operations.
All operations but the $\textsc{prove}(x)$ that return $true$ are linearized in a point of their
execution interval.

Thus, it remains to consider the $\textsc{prove}(x)$ operations that return $true$. They are linearized before the first $\textsc{append}(x)$ in the linearization. Since, they read $true$ in all low-level auditable registers, there is no $\textsc{append}(x)$ that completes before any of these $\textsc{prove}(x)$
is invoked, otherwise one of the low-level read would have return $\mathit{false}$.

If $\textsc{prove}(x)$ $op$ completes before a $\textsc{read}()$ $op'$ starts, then the last step of $op$ precedes the linearization point of $op'$. Thus, $op$ is linearized before $op'$. Similarly if $op$ is invoked after $op'$ completes, then the linearization point of $op'$ precedes the last step of $op$ which is then linearized after $op'$. 

This concludes the proof since $\textsc{prove}(x)$ are linearized in the order of their last step.


\section{Conclusions and Future Work}
\label{sec:conclusions}

In this work, we extended the concept of auditability from single-writer registers to more general shared objects. We start by providing a rigorous characterization of the synchronization power required to support auditable multi-writer registers. Our results establish a tight bound on the consensus number necessary for achieving auditability, demonstrating the feasibility of implementing auditable storage mechanisms.

Looking ahead, there are several promising directions for future research. First, extending auditability to a broader range of shared objects beyond registers and LL/SC remains an open challenge. Second, investigating the impact of adversarial behavior on auditable implementations could lead to more robust security guarantees. 
Finally, exploring efficient, scalable implementations of auditable objects in real-world distributed storage systems could bridge the gap between theoretical feasibility and practical deployment.
In particular, while sliding registers are not supported in hardware
(to the best of our knowledge), they bear resemblance to \emph{shift registers}. 
This hardware object~\cite{IntelShift} 
holds the last $w$ bits ``shifted in'' the register
and has a functionality similar to a sliding register; 
the consensus number of shift registers is $w$~\cite{Aspnes2025}, 
Although they cannot be used as a direct replacement of sliding registers in our implementations, our algorithmic insights might be leveraged to develop other algorithms that employ shift registers.

\paragraph*{Acknowledgments:}
Hagit Attiya is supported by the Israel Science Foundation (22/1425 and 25/1849).
Antonio Fernández Anta is supported by project PID2022-140560OB-I00 (DRONAC) funded by MICIU/AEI /10.13039/501100011033 and ERDF, EU.
Alessia Milani is supported in part by ANR project TRUSTINCloudS (ANR-23-PECL-0009).
Corentin Travers is supported in part by ANR project DUCAT (ANR-20-CE48-0006).

\bibliography{refs}
\end{document}